\documentclass[twocolumn,twoside]{IEEEtran}
\usepackage{times} 
\usepackage{epsfig}
\usepackage{xtab,diagbox}
\usepackage{arydshln}
\usepackage{graphicx}
\usepackage{epstopdf}
\usepackage{color,subfigure, latexsym, amsmath, amsfonts, amssymb,cite,mathtools}
\usepackage{algorithm}
\usepackage{algorithmic}
\usepackage{epsfig}
\usepackage{graphicx}
\usepackage{epstopdf}
\usepackage{enumitem}
\usepackage{bm,amsmath,amssymb,amsfonts,graphicx,epsfig,amsthm,color, xcolor}
\usepackage{bbold,dsfont}
\usepackage{float}
\usepackage[hyphens]{url}
\PassOptionsToPackage{bookmarks=false}{hyperref}
\usepackage{mathrsfs}
\usepackage[depth=-1]{bookmark}
\usepackage{booktabs}       
\usepackage{amsfonts}       
\usepackage{microtype}      
\usepackage{placeins}
\usepackage{times}
\usepackage{graphicx} 

\usepackage{wrapfig}



\usepackage{accents}
\newcommand{\R}{\mathbb{R}}

\newcommand{\K}{\mathcal{K}}

\newcommand{\Q}{\mathcal{Q}}
 
\newcommand{\C}{\mathbb{C}}

\newtheorem{lemma}{Lemma}

\newtheorem{theorem}{Theorem}
\newtheorem{definition}{Definition}
\newtheorem{assumption}{Assumption}
\newtheorem{problem}{Problem}

\theoremstyle{remark}\newtheorem{remark}{Remark}

\allowdisplaybreaks

\title{\LARGE \bf Data-Driven Frequency-Selective Output Regulation of Nonlinear Systems under Almost Periodic Exosignals
}

\author{ Yifei Li, Wenjie~Liu, Gang Wang,~\IEEEmembership{Senior Member,~IEEE}, and Lihua Xie,~\IEEEmembership{Fellow,~IEEE}
		 \thanks{Yifei Li and Wenjie Liu are with the State Key Lab of Autonomous Intelligent Unmanned Systems, Beijing Institute of Technology, Beijing 100081, China, and also also with the Centre for Advanced Robotics Technology  Innovation (CARTIN), School of Electrical and Electronic Engineering, Nanyang Technological University, Singapore (e-mail: li.yifei@ntu.edu.sg; wenjie.liu@ntu.edu.sg).}
         \thanks{Gang Wang is with the State Key Lab of Autonomous Intelligent Unmanned Systems, Beijing Institute of Technology, Beijing 100081, China (e-mail: gangwang@bit.edu.cn).}
        \thanks{Lihua Xie is with the Centre for Advanced Robotics Technology Innovation (CARTIN), School of Electrical and Electronic Engineering, Nanyang Technological University, Singapore (e-mail: elhxie@ntu.edu.sg).

			}
	}

\begin{document}
	
	\maketitle
	

    \allowdisplaybreaks
	
	\begin{abstract}
This paper studies output regulation for a class of unknown continuous-time nonlinear systems driven by almost periodic exosignals. The plant dynamics are assumed to be linearly parameterized over a prescribed nonlinear dictionary, while all coefficient matrices in the plant, input channel, output map, and exosignal channel are unknown. Since the plant model is unavailable, exact nonlinear output regulation would generally require model identification followed by the solution of nonlinear regulator equations. To avoid these steps, we pursue a frequency-selective regulation objective: the steady-state regulation error is allowed to be almost periodic, but its Fourier--Bohr coefficients at prescribed exosystem frequencies are guaranteed to vanish, and the residual error energy is explicitly bounded. To this end, a $p$-copy internal model is embedded in a dynamic controller, yielding an augmented nonlinear system whose unknown constant matrices are represented directly by measured data. A noise-robust semidefinite program is derived to synthesize the controller gain without model identification and without measuring the exosignal amplitudes or phases. The resulting closed-loop vector field is made exponentially contractive on a prescribed operating set, which implies the existence and uniqueness of a bounded and attracting trajectory. By combining contraction theory with Fourier--Bohr analysis, we prove that this steady-state trajectory is almost periodic, that the embedded-frequency components of the regulation error are eliminated, and that the unmodeled spectral components satisfy a Parseval-type time-averaged energy bound. Numerical and physics-based simulations on a quadrotor with a cable-suspended payload illustrate the effectiveness of the proposed data-driven internal-model design.
\end{abstract}


\section{Introduction}
Output regulation is a central problem in control theory. Its objective is to design a feedback controller such that the regulated output tracks reference signals and rejects disturbances generated by an autonomous exosystem. For linear systems, the regulator equations and the internal model principle provide complete and elegant solutions. For nonlinear systems, however, the corresponding nonlinear output regulator equations (NOREs) are generally difficult to solve, even when a precise system model is available. The difficulty becomes more pronounced in data-driven settings, where the system matrices, input map, disturbance channel, and output map are not known \emph{a priori}.

A large body of nonlinear output regulation theory has been developed for model-based systems; see, e.g., \cite{huang2004nonlinear,Huang2004general,Khan2020,Farnaz2016}. These results provide strong asymptotic regulation guarantees, but they usually require explicit plant models and the solvability of the NOREs. For practical systems such as aerial vehicles, robotic platforms, and power converters, accurate parametric models are often unavailable, and the dominant disturbances may contain several incommensurate oscillatory components. This motivates regulation formulations that are both data-driven and compatible with almost periodic exogenous signals.

Data-driven control has recently become a powerful alternative to model-identification-based synthesis. In the linear case, behavioral methods based on Willems et al's fundamental lemma \cite{willems2005note} have led to data-enabled predictive control, robust control, aperiodic control, and output regulation designs directly from measured trajectories \cite{35Persis2020,Waarde2022,you2023auto,Eising2025,Rueda2022,Coulson2019data,Wildhagen2023,Qi2023data,berberich2019data,Zhu2024,wang2026data,Chang2024localized}. Reinforcement-learning and adaptive-dynamic-programming methods have also been used for optimal output regulation \cite{Lewis2009,Lewis2012,Gao2016,Wu2022,lin2024ddimp}. For nonlinear systems, existing data-driven regulation results remain limited. Harmonic regulation methods eliminate selected Fourier components under periodic exosignals \cite{hu2024output}; incremental passivity and local-linearization approaches provide alternative data-driven conditions under additional structural assumptions \cite{Liu2025,Liu2025regulation}. These works demonstrate the feasibility of data-driven output regulation, but they do not yet provide a direct synthesis method for nonlinear systems under general finite-spectrum almost periodic exosignals with noisy data and an explicit steady-state spectral characterization.

This paper develops such a framework for a class of nonlinear systems whose vector field is linearly parameterized over a known dictionary of continuously differentiable basis functions. The coefficient matrices multiplying the dictionary, the input, the exosignal, and the regulated output are unknown. The exosystem is neutrally stable with semisimple imaginary-axis eigenvalues, so its trajectories are finite sums of constants and sinusoidal signals. When the frequencies are rationally independent, the resulting signal is almost periodic but not periodic. Instead of requiring exact asymptotic regulation through the NOREs, we formulate a frequency-selective regulation objective: the controller must eliminate the Fourier--Bohr coefficients of the steady-state error at the embedded exosystem frequencies, while the remaining spectral components are quantified by a time-averaged energy bound.

The proposed controller combines an internal model with direct data-driven stabilization. A $p$-copy internal model embeds the constant mode and the selected oscillatory modes into a dynamic compensator. The plant and compensator are then written as an augmented nonlinear system. By exploiting the known dictionary structure and the known frequency basis of the exosystem, the unknown matrices in the augmented dynamics are represented by offline input-state data, while the unknown exosignal amplitudes and phases are eliminated through linear annihilation constraints. A robust semidefinite program (SDP) is derived to synthesize a feedback gain directly from noisy data. The SDP enforces an incremental dissipation inequality, which yields exponential contraction of the closed-loop dynamics on a prescribed operating region.

The use of contraction is important for two reasons. First, it avoids solving the NOREs and instead guarantees that all trajectories converge exponentially to a unique bounded solution driven by the exosignal. Second, since the exosignal is almost periodic and the closed-loop system is input-to-state convergent, the attracting steady-state solution is also almost periodic. This enables a frequency-domain analysis based on Fourier--Bohr coefficients. The boundedness of the internal-model state implies that the Fourier--Bohr coefficients of the steady-state error at the embedded frequencies vanish. A Parseval-type argument further bounds the residual error energy in terms of the derivative bound of the error and the smallest magnitude of the unmodeled frequencies.

The main contributions are summarized as follows.
\begin{itemize}
\item A frequency-selective data-driven output regulation problem is formulated for nonlinear systems under almost periodic exosignals with finite spectrum. The formulation avoids the NOREs and characterizes regulation through Fourier--Bohr coefficients of the steady-state error.
\item A direct noisy-data synthesis method is developed for an internal-model-based nonlinear regulator. The unknown plant, input, output, and exosignal-channel matrices are not identified; instead, a robust SDP constructs a gain that certifies contraction of the augmented closed-loop system.
\item A complete steady-state analysis is provided. Exponential contraction implies the existence, uniqueness, and exponential attractivity of an almost periodic steady-state trajectory. The embedded-frequency components of the regulation error are eliminated, and the residual frequency components are bounded in time-averaged energy.
\item Numerical and physics-based quadrotor simulations show that the proposed controller attenuates payload-induced oscillatory disturbances more effectively than a comparable data-driven controller without the internal model.
\end{itemize}

The remainder of the paper is organized as follows. Section~\ref{sec:pre} introduces almost periodic functions and formulates the problem. Section~\ref{sec:result} presents the internal-model-based controller, the data-driven SDP, and the regulation analysis. Section~\ref{sec:simu} validates the method on a quadrotor with a suspended payload. Section~\ref{sec:conclusion} concludes the paper.

\emph{Notation.} Let $\mathbb{R}$, $\mathbb{R}_{>0}$, and $\mathbb{N}$ denote the sets of real numbers, positive real numbers, and non-negative integers, respectively. 
The sets of real symmetric positive definite matrices in $\mathbb{R}^{n\times n}$ are written as $\mathbb{S}^n_{>0}$.
For a matrix $M$, $M \succ (\succeq)0$ and $M \prec (\preceq)0$ denote positive (semi-)definiteness and negative (semi-)definiteness, respectively. The spectral norm of a matrix $M$ is denoted by $\|M\|$, and the Euclidean norm of a vector $x \in \mathbb{R}^n$ by $\|x\|$. 
For a symmetric matrix $M$, $\lambda_{\min}(M)$ denote the smallest eigenvalue of $M$.
For a series of vectors $x_1,\ldots,x_N$, define ${\rm col}(x_1,\ldots,x_N) = [x_1^\top,\ldots,x_N^\top]^\top$. The operators ${\rm diag}\{\cdot\}$ and ${\rm blkdiag}\{\cdot\}$ denote diagonal and block diagonal matrices, respectively. The symbol $I_N$ denotes the $N$-dimensional identity matrix and $1_N$ the $N$-dimensional column vector of ones. In addition, the symbol $\star$ denote the symmetric part in symmetric matrices.

\section{Preliminaries and Problem Formulation}\label{sec:pre}
This section begins with a brief review of almost periodic functions and then formulates the considered nonlinear output regulation problem.

\subsection{Almost Periodic Functions}\label{sec:apf}

We first briefly recall several definitions and fundamental properties of almost periodic functions that will be used in the subsequent analysis.

\begin{definition}[{\cite[Theorem 1.11]{corduneanu2009almost}}]\label{def:apf}
A continuous function $f:\mathbb{R}\to\mathbb{R}^n$ is said to be almost periodic if for any $\varepsilon>0$, there exists a positive constant $l(\varepsilon)>0$ such that every interval of the real line of length $l(\varepsilon)$ contains at least one number $\tau$ satisfying
\[
\|f(t+\tau)-f(t)\|<\varepsilon,\quad \forall t\in\mathbb{R}.
\]
\end{definition}

We next recall several important properties of almost periodic functions, including the mean value, Fourier--Bohr coefficients, and a Parseval-type identity.

\begin{lemma}[{\cite[Theorem 1.12]{corduneanu2009almost}}]\label{lem:mean}
Let $f:\mathbb{R}\to\mathbb{R}^n$ be an almost periodic function. Then the limit
\[
M\{f\}:=\lim_{T\to\infty}\frac{1}{T}\int_{a}^{a+T} f(t)\,dt
\]
exists uniformly with respect to $a\in\mathbb{R}$. Moreover, the value $M\{f\}$ is independent of $a$ and is called the mean value of $f$.
\end{lemma}

For an almost periodic function $f$, its Fourier--Bohr coefficient associated with frequency $\lambda\in\mathbb{R}$ is defined as
\[
\hat f(\lambda):=\lim_{T\to\infty}\frac{1}{T}\int_{a}^{a+T} f(t)e^{-i\lambda t}\,dt,
\]
where $i$ is the imaginary unit and the limit exists uniformly with respect to $a\in\mathbb{R}$.

The set of frequencies corresponding to nonzero Fourier--Bohr coefficients is called the spectrum of $f$. Denote the spectrum by $\{\sigma_k\}_{k=1}^{\infty}$ and let $A_k := \hat f(\sigma_k).$
Then $f$ admits the Fourier expansion
\[
f(t)= \sum_{k=1}^{\infty} A_k e^{i\sigma_k t}.
\]

The following Parseval's equation will be used in the analysis.
\begin{lemma}[{\cite[Theorem 1.18]{corduneanu2009almost}}]\label{lem:parseval}
Let $f:\mathbb{R}\to\mathbb{R}^n$ be an almost periodic function with Fourier--Bohr coefficients $\{A_k\}$. Then Parseval's equation holds
\[
\sum_{k=1}^{\infty} \|A_k\|^2 = M\{\|f\|^2\}.
\]
\end{lemma}

With these preliminaries in place, we are now ready to formulate the problem considered in this paper.


\subsection{Problem Formulation}

Consider a class of continuous-time nonlinear systems described by
\begin{subequations}\label{eq:sys:non}
	\begin{align}
		\dot{x} & = AF(x)+Bu+Ev \\
		e &= CF(x)+Hv \label{eq:sys:non:e}
	\end{align}
\end{subequations}
where $x\in \mathbb{R}^{n}$ is the state, $u\in \mathbb{R}^{m}$ is the control input, and $v \in\mathbb{R}^{q}$ consists of the reference signal to be tracked and the disturbance to be rejected, which is generated by 
\begin{equation}\label{eq:sys:non:exo} 
    \dot{v}  = S v. 
\end{equation}
In particular, $e \in \mathbb{R}^{p}$ is the regulation error, which represents the mismatch between the system output $y=CF(x)$ and the output of the exosystem $y_v=-Hv$.

We express the drift vector field as a linear combination of known functions in $F: \R^n\rightarrow\R^{n_f}$, which includes linear and nonlinear functions, i.e.,
$$F(x)=\left[\begin{matrix}
    x\\Q(x)
\end{matrix}\right],$$ 
where $Q: \R^{n}\rightarrow\R^{n_f-n}$ contains only nonlinear functions. All the functions of $F(x)$ are continuously differentiable. 
In the considered data-driven setting, the matrices $A$, $B$, $C$, $E$, and $H$ are unknown, and the exosignal $v$ is not measurable.

In the classical output regulation framework, \eqref{eq:sys:non:exo} is referred to as the exosystem. To characterize its properties, we impose the following assumption.
\begin{assumption}\label{as:exo:non}
	The matrix $S$ in \eqref{eq:sys:non:exo} is known and all its eigenvalues are semisimple with zero real parts.
\end{assumption}

Under Assumption~\ref{as:exo:non}, the solution of the exosystem \eqref{eq:sys:non:exo} consists of constant and sinusoidal components with arbitrary amplitudes and initial phases, whose frequencies are determined by the eigenvalues of $S$. In particular, the exosignal $v(t)$ can be expressed as a finite linear combination of constant and sinusoidal functions with known frequencies. Therefore, it can be concluded that $v(t)$ is an almost periodic function in the sense of Definition~\ref{def:apf}. Moreover, such signals capture a wide range of practically relevant references and disturbances, including multi-frequency excitations whose amplitudes depend on the unknown initial condition of the exosystem.

Furthermore, it follows from Assumption~\ref{as:exo:non} that the exosystem~\eqref{eq:sys:non:exo} admits the following decomposition. In particular, there exist positive integers $r$ and $s$ satisfying $r+2s=q$, and known frequencies $\sigma_1,\ldots,\sigma_s\in\mathbb{R}_{>0}$ such that
\[
v_i(t)=v_i(0), \quad i=1,\ldots,r,
\]
and for each $k=1,\ldots,s$, the block component $v_{r+k}(t)\in\mathbb{R}^2$ satisfies
\[
v_i(t)=\underbrace{\begin{bmatrix}
   v_{r+k,1}(0) & v_{r+k,2}(0)\\
   v_{r+k,2}(0) & -v_{r+k,1}(0)
\end{bmatrix}}_{:=L_{r+k}(v(0))}
\underbrace{\begin{bmatrix}
      \cos(\sigma_k t)\\
       \sin(\sigma_k t)
\end{bmatrix}}_{:=M_{r+k}(t)}.
\]
Consequently, the exosignal $v(t)$ can be expressed in the compact form
\begin{equation}\label{eq:decompose}
    v(t)=LM(t),
\end{equation}
where
\begin{align*}
&L = {\rm blkdiag}\big\{
v_1(0),\ldots,v_r(0),\\
&\qquad\qquad\qquad L_{r+1}(v(0)),\ldots,L_{r+s}(v(0))
\big\},\\
&M(t) = {\rm col}\big(
1_r, M_{r+1}(t),\ldots,M_{r+s}(t)
\big).
\end{align*}
Here, $M(t)$ consists of known basis functions determined by the frequencies $\{\sigma_k\}$, while $L$ depends on the unknown initial condition $v(0)$.

In classical output regulation, the objective is to drive the regulation error $e(t)$ to zero. However, in nonlinear and data-driven settings, such a requirement is often difficult to achieve. This is because achieving exact regulation typically relies on solving the so-called NOREs, which seek sufficiently smooth mappings $\pi(v)$ and $\varrho(v)$ satisfying $\pi(0)=0$ and $\varrho(0)=0$ such that
\begin{subequations}\label{eq:NORE}
\begin{align}
\frac{\partial \pi(v)}{\partial v}Sv
&=
AF({\pi}(v))+B{\varrho}(v)+Ev,
\\
0
&=
CF(\pi(v))+Hv.
\end{align}
\end{subequations}
However, solving the NOREs \eqref{eq:NORE} is generally difficult even when the system model is fully known, and becomes intractable in the data-driven setting considered here, where the system dynamics are unknown. 
Motivated by this observation, we do not pursue exact regulation via explicit solutions of the NOREs. Instead, we exploit the internal model principle~\cite{huang2004nonlinear}, which embeds the exosystem dynamics into the controller and reformulates the regulation problem as a stabilization problem of an augmented system.
However, existing data-driven harmonic regulation methods for nonlinear systems, e.g., \cite{hu2024output}, typically impose restrictive assumptions on the exosystem, such as requiring rationally related frequencies, thereby limiting the admissible signals to periodic ones. In contrast, under Assumption~\ref{as:exo:non}, the exosignal is allowed to be almost periodic, which enlarges the class of admissible signals but also introduces additional challenges.
Accordingly, rather than aiming at exact asymptotic regulation, this paper adopts a frequency-domain perspective and studies a data-driven output regulation problem for nonlinear systems subject to almost periodic exosignals, where selected frequency components of the steady-state error are eliminated while the residual part is quantitatively characterized. We next formally state the problem.

\begin{problem}\label{pro}
Given the nonlinear system \eqref{eq:sys:non}, the known frequency set $\{0,\sigma_1,\ldots,\sigma_s\}$ of the exosystem, and offline measured input--state--error data, design a dynamic feedback controller without identifying the matrices $A,B,C,E,H$ and without measuring the exosignal amplitudes such that: i) all closed-loop trajectories remain bounded and converge to a unique almost periodic steady-state response; ii) the Fourier--Bohr coefficients of the steady-state regulation error at $0,\sigma_1,\ldots,\sigma_s$ are zero; and iii) the residual spectral content of the steady-state error is quantitatively bounded.
\end{problem}

\section{Data-Driven Regulator Design and Analysis}\label{sec:result}
In this section, we develop a data-driven internal model-based controller for the nonlinear system \eqref{eq:sys:non}, thereby solving Problem~\ref{pro} directly from noisy data.
We first introduce the regulator structure and the resulting closed-loop system, and then present the data-driven synthesis method together with the stability and regulation analysis.

\subsection{Regulator Structure and Closed-loop System}

We first specify the dynamic regulator. The construction follows the internal model principle, but the role of the internal model here is frequency-selective: it embeds only the constant mode and the prescribed sinusoidal modes of the exosystem. Consequently, the controller is not claimed to solve the full NOREs. Instead, the embedded modes are used to cancel the corresponding Fourier--Bohr coefficients of the steady-state regulation error.

Define $\mathcal{F}(x,z)=\left[\begin{smallmatrix}
    x\\z\\Q(x)
\end{smallmatrix}\right]$. The controller is given by
\begin{subequations}\label{eq:ctrl}
	\begin{align}
		u &= \mathcal{K}\mathcal{F}(x,z)\label{eq:ctrl:u}\\
		\dot{z} &= \Phi z + \Gamma e \label{eq:inter_model}
	\end{align}
\end{subequations}
where the gain matrix $\mathcal K$ is to be designed and $z\in\mathbb{R}^{n_z}$ denotes the internal model state. The matrices $(\Phi,\Gamma)$ are constructed as follows
\begin{equation}\label{eq:inter_matrices}
    \Phi={\rm blkdiag}\{\underbrace{\phi,\ldots,\phi}_{p~\text{times}}\},~~\Gamma={\rm blkdiag}\{\underbrace{G,\ldots,G}_{p~\text{times}}\}
\end{equation}
with
\begin{align*}
      \phi &= {\rm blkdiag}\left\{0,\begin{bmatrix}
    0 & \sigma_1\\
    -\sigma_1 & 0
\end{bmatrix},\ldots,\begin{bmatrix}
    0 & \sigma_s\\
    -\sigma_s & 0
\end{bmatrix}\right\}\\
G &= {\rm col}\left( \gamma,
\bar G_1,\ldots,\bar G_s \right),
\end{align*}
where $\gamma>0$ and $\bar G_k\in\mathbb{R}^2\setminus\{0\}$ are chosen such that each pair $\left(\left[\begin{smallmatrix}
0 & \sigma_k\\
-\sigma_k & 0
\end{smallmatrix}\right],\bar{G}_k\right),\quad k=1,\ldots,s,$ is controllable.
By construction, the matrix $\Phi$ embeds the frequency set
$\{0,\sigma_1,\ldots,\sigma_s\}$ associated with the exosystem. This  enables the controller to reproduce signals at these frequencies, which plays a key role in shaping the steady-state regulation error in the frequency domain.

To facilitate the subsequent data-driven design, we reorganize the system representation by partitioning the matrices $A$ and $C$ according to the structure of $F(x)$ as
$\left[\begin{smallmatrix}
    A\\C
\end{smallmatrix}\right]
=\left[\begin{smallmatrix}
    \bar{A} &\hat{A}\\ \bar{C} &\hat{C}
\end{smallmatrix}\right]$.
Then, combining \eqref{eq:sys:non} and \eqref{eq:inter_model}, the resulting augmented system can be expressed as
\begin{subequations}\label{eq:aug:sys}
    \begin{align}
        \begin{bmatrix}
            \dot x\\ \dot z
        \end{bmatrix} &= {A_\xi} \mathcal{F}(x,z) + {B_\xi}u + {E_\xi}v,\label{eq:aug:sys:x}\\
        e&=C_{\xi}\mathcal{F}(x,z)+Hv,\label{eq:aug:sys:e}
    \end{align}
\end{subequations}
where
\begin{align*}
    A_\xi&=\begin{bmatrix}
            \bar{A} & 0 & \hat{A}\\
            \Gamma \bar{C} & \Phi & \Gamma \hat{C}
        \end{bmatrix}, ~B_\xi=\begin{bmatrix}
            B\\0
        \end{bmatrix},\\
        E_\xi&=\begin{bmatrix}
            E\\ \Gamma H
        \end{bmatrix},\quad ~C_{\xi}=\begin{bmatrix}
            \bar C & 0 & \hat C
        \end{bmatrix}.
\end{align*}
Under the controller \eqref{eq:ctrl}, the augmented system \eqref{eq:aug:sys:x} takes the closed-loop form
\begin{equation}\label{eq:closeloop:model}
\begin{bmatrix}
\dot x\\ \dot z
\end{bmatrix}
=
\big(A_\xi + B_\xi \mathcal{K}\big)\mathcal{F}(x,z) + E_\xi v.
\end{equation}

\subsection{Data-Driven Regulator Design}

We now derive a direct data-driven synthesis condition. During the offline experiment the controller state $z$ is generated from the measured regulation error $e$ through \eqref{eq:inter_model}, while the input $u$ is chosen to be sufficiently exciting. The exosignal $v$ itself is not assumed to be measured. For each sampling interval $\tau\in[t_\ell,t_{\ell+1})$, the derivatives are approximated by forward differences, namely $\dot x (\tau):=\frac{x(t_{\ell+1})-x(t_{\ell})}{t_{\ell+1}-t_{\ell}}$ and $\dot z (\tau):=\frac{z(t_{\ell+1})-z(t_{\ell})}{t_{\ell+1}-t_{\ell}}$. The resulting approximation error, including measurement noise and finite-difference error, is collected in $d(\tau)$. Therefore, for $\tau=0,\ldots,T-1$, the augmented dynamics satisfy
\begin{equation}\label{eq:noise-data-xi}
     \begin{bmatrix}
            \dot x(\tau)\\ \dot z(\tau)
        \end{bmatrix} ={A_\xi} \mathcal{F}(x(\tau),z(\tau)) + {B_\xi}u(\tau) + {E_\xi}v(\tau) + d(\tau),
\end{equation}
where $d(\tau) \in \mathbb{R}^{n_{\xi}}$ with $n_{\xi}=n+n_z$ denotes the unknown approximation error induced by the finite-difference approximation and possible measurement noise.

By stacking the collected data, we define the following data matrices
\begin{subequations}\label{eq:data:lti}
	\begin{align}
		U_- & = \left[\begin{matrix}
			{u}(0) & {u}(1) & \cdots & {u}(T-1)
		\end{matrix}\right] \label{eq:data:lti:U-}\\
  \Xi_- &= \left[\begin{matrix}
			{x}(0) & {x}(1) & \cdots & {x}(T-1)\\
            {z}(0) & {z}(1) & \cdots & {z}(T-1)\\
            Q(x(0)) &  Q(x(1)) & \cdots &  Q(x(T-1))\\
		\end{matrix}\right] \label{eq:data:lti:Xi-}\\
  \Xi_+ &= \left[\begin{matrix}
			\dot{x}(0) & \dot{x}(1) & \cdots & \dot{x}(T-1)\\
            \dot{z}(0) & \dot{z}(1) & \cdots & \dot{z}(T-1)
		\end{matrix}\right]\label{eq:data:lti:Xi+}.
	\end{align}
Similarly, define the data matrices associated with the disturbance and exosignal as
    \begin{align}
	D_- &= \left[\begin{matrix}
		{d}(0) & {d}(1) & \cdots & {d}(T-1)
	\end{matrix}\right] \label{eq:data:lti:D}\\
	V_- &= \left[\begin{matrix}
		{v}(0) & {v}(1) & \cdots & {v}(T-1)
	\end{matrix}\right] \label{eq:data:lti:V}.
\end{align}
\end{subequations}
According to \eqref{eq:noise-data-xi}, the above data matrices satisfy 
\begin{equation}\label{eq:relation}
    \Xi_+=A_{\xi}\Xi_-+B_{\xi}U_-+E_{\xi}V_-+D_-.
\end{equation}

From the decomposition of the exosystem established in \eqref{eq:decompose}, the exosignal data matrix $V_-$ admits the factorization $V_-=LM_-$, where $M_-=\begin{bmatrix}
    M(0) & M(1) & \ldots & M(T-1)
\end{bmatrix}$.
Substituting this into \eqref{eq:relation} yields
\begin{equation}\label{eq:data:relation}
\Xi_+ = A_{\xi}\Xi_- + B_{\xi}U_- + E_{\xi}L M_- + D_-.
\end{equation}
We note that the matrix $L$ depends on the unknown initial condition of the exosystem and is therefore unknown, whereas $M_-$ is completely determined by the known frequency set $\{\sigma_k\}$ and is thus known.

Next, to characterize the effect of the unknown disturbance $D_-$, we impose the following assumption.
\begin{assumption}\label{as:D_bound}
    A matrix $\Theta\succeq 0$ is available such that the stacked differentiation and measurement error satisfies
    \begin{equation}
        D_-\in\mathcal{D}:=\{D\in\R^{n_{\xi}\times T}: DD^\top\preceq \Theta \Theta^\top\}.
    \end{equation}
    The matrix $\Theta$ may be selected from sensor specifications, repeated experiments, or a conservative a priori bound on the derivative-estimation error.
\end{assumption}

The following lemma synthesizes a controller gain $\K$  directly from noisy data and enforces exponential contraction of the resulting closed-loop system.

\begin{lemma}\label{lem:contraction}
Consider the augmented system \eqref{eq:aug:sys} under the controller \eqref{eq:ctrl}. 
Assume that there exists a known matrix $R_Q\in\mathbb{R}^{n\times r}$ and a set $\mathcal{X}\subseteq\R^{n}$ such that
\begin{equation}\label{eq:nonlin:bound}
\frac{\partial Q}{\partial x}(x)^\top \frac{\partial Q}{\partial x}(x)
\preceq R_Q R_Q^\top, \quad \forall x\in\mathcal{X}.
\end{equation}
Define $\mathcal{R}_Q := \left[\begin{smallmatrix} R_Q \\ 0 \end{smallmatrix}\right].$
Suppose there exist matrices $P\in\mathbb{S}^{n_\xi}_{>0}$, 
$Y\in\mathbb{R}^{T\times n_\xi}$, $\hat{G}\in\mathbb{R}^{T\times (n_f-n)}$, 
and scalars $\alpha>0$, $\mu>0$ satisfying
\begin{subequations}\label{eq:sdp}
\begin{align}
&\Xi_- Y = \begin{bmatrix} P \\ 0 \end{bmatrix}, \label{eq:sdp:1}\\
&\Xi_- \hat{G} = \begin{bmatrix} 0 \\ I_{n_f-n} \end{bmatrix}, \label{eq:sdp:2}\\
&M_- \begin{bmatrix} Y & \hat{G} \end{bmatrix} = 0, \label{eq:sdp:3}\\
&\begin{bmatrix}
\Upsilon(Y,\alpha,\mu)& \Xi_+ \hat{G} & P\mathcal{R}_Q & Y^\top \\
\star & -I_{n_f-n} & 0 & \hat{G}^\top \\
\star & \star & -I_r & 0 \\
\star & \star & \star & -\mu I_T
\end{bmatrix}
\preceq 0,
\label{eq:sdp:4}
\end{align}
\end{subequations}
where $\Upsilon(Y,\alpha,\mu):=\Xi_+ Y + (\Xi_+ Y)^\top + \alpha I_{n_\xi} + \mu \Theta\Theta^\top$.
Then the controller \eqref{eq:ctrl} with
\begin{equation}\label{eq:K}
\mathcal{K} = U_- \begin{bmatrix} YP^{-1} & \hat{G} \end{bmatrix}
\end{equation}
renders the closed-loop system exponentially contractive on the operating set $\mathcal X\times\mathbb{R}^{n_z}$, with respect to the constant metric $P^{-1}$, for all admissible errors $D_-\in\mathcal D$.
\end{lemma}

\begin{proof}
    We first exploit the data relation \eqref{eq:data:relation} to derive a data-driven representation of the closed-loop system \eqref{eq:closeloop:model}.

    Introduce a matrix $G=\begin{bmatrix}\bar{G} & \hat{G}\end{bmatrix}$ with $\bar{G}\in\mathbb{R}^{T\times n_\xi}$ and $\hat{G}\in\mathbb{R}^{T\times (n_f-n)}$. 
    If the SDP \eqref{eq:sdp} is feasible and we set $Y=\bar{G}P$, then it follows from \eqref{eq:sdp:1}--\eqref{eq:sdp:3} and \eqref{eq:K} that
    \begin{equation}\label{eq:LUM=IK0}
        \begin{bmatrix}
        \mathcal K\\
        I\\
        0
        \end{bmatrix}
        =
        \begin{bmatrix}
        U_-\\
        \Xi_-\\
        M_-
        \end{bmatrix} G.
    \end{equation}

    Pre-multiplying $\begin{bmatrix} B_{\xi} & A_{\xi} & E_{\xi}L \end{bmatrix}$ on both sides of \eqref{eq:LUM=IK0} yields
    \begin{align}\label{eq:xiG}
        A_{\xi} + B_{\xi} \mathcal K
        = \begin{bmatrix} B_{\xi} & A_{\xi} & E_{\xi}L \end{bmatrix}
        \begin{bmatrix}
        U_-\\
        \Xi_-\\
        M_-
        \end{bmatrix} G \overset{\eqref{eq:data:relation}}{=} (\Xi_{+}-D_-) G.
    \end{align}
    Substituting into the closed-loop system~\eqref{eq:closeloop:model} leads to the data-dependent dynamics
    \begin{equation}\label{eq:closeloop:data}
        \begin{aligned}
        \begin{bmatrix}
        \dot x\\ \dot z
        \end{bmatrix}
        &=(\Xi_{+}-D_-) G \mathcal{F}(x,z) + E_\xi v\\
        &=(\Xi_{+}-D_-) \bar G \begin{bmatrix} x\\ z \end{bmatrix}
        +(\Xi_{+}-D_-) \hat G Q(x)+ E_\xi v.
        \end{aligned}
    \end{equation}

    We now proceed to show the contraction of the closed-loop system. By the Schur complement, \eqref{eq:sdp:4} is equivalent to
    \begin{equation}\label{eq:lmi_splitted}
        \begin{aligned}
            &\begin{bmatrix}
            \Xi_+ Y + (\Xi_+ Y)^\top + \alpha I_{n_{\xi}} & \Xi_+ \hat{G} & P \mathcal{R}_Q \\
            \star & -I_{n_f-n} & 0 \\
            \star & \star & -I_r
            \end{bmatrix} \\
            &+ \mu
            \begin{bmatrix}
            - I \\ 0 \\ 0
            \end{bmatrix}
            \Theta \Theta^\top
            \begin{bmatrix}
            - I \\ 0 \\ 0
            \end{bmatrix}^\top
            + \mu^{-1}
            \begin{bmatrix}
            Y^\top \\ \hat{G}^\top \\ 0
            \end{bmatrix}
            \begin{bmatrix}
            Y^\top \\ \hat{G}^\top \\ 0
            \end{bmatrix}^\top
            \preceq 0 .
        \end{aligned}
    \end{equation}

    Applying Petersen's lemma \cite{bisoffi2022data} to eliminate $D$ satisfying $DD^\top\preceq\Theta\Theta^\top$, we obtain from \eqref{eq:lmi_splitted} that for all admissible $D$,
    \begin{equation}\label{eq:lmi_with_D}
    \begin{aligned}
        &
        \begin{bmatrix}
        \Xi_+ Y + (\Xi_+ Y)^\top + \alpha I_{n_\xi}
        & \Xi_+ \hat{G}
        & P \mathcal{R}_Q \\
        (\Xi_+ \hat{G})^\top
        & -I_{n_f-n}
        & 0 \\
        \mathcal{R}_Q^\top P^\top
        & 0
        & -I_r
        \end{bmatrix} \\
       & + \begin{bmatrix}
        - I \\ 0 \\ 0
        \end{bmatrix}
        D
        \begin{bmatrix}
        Y & \hat{G} & 0
        \end{bmatrix}
        + \begin{bmatrix}
        Y^\top \\ \hat{G}^\top \\ 0
        \end{bmatrix}
        D^\top
        \begin{bmatrix}
        - I & 0 & 0
        \end{bmatrix} \preceq 0,
        \end{aligned}
    \end{equation}
    which implies
    \begin{equation}
        \label{eq:lmi_rewritten}
        \begin{bmatrix}
        \widetilde{\Upsilon}(Y,\alpha)
        & (\Xi_+ -  D) \hat{G}
        & P \mathcal{R}_Q \\
       \star
        & -I_{n_f-n}
        & 0 \\
        \star
        & \star
        & -I_r
        \end{bmatrix}
        \preceq 0,
        \quad \forall D \in \mathcal{D}.
    \end{equation}
    with $\widetilde{\Upsilon}(Y,\alpha):=(\Xi_+ -  D) Y
        + Y^\top (\Xi_+ - D)^\top
        + \alpha I_{n_\xi}$.
    
    Applying the Schur complement yields
    \begin{equation}
        \begin{bmatrix}
        \widetilde{\Upsilon}(Y,\alpha) + P \mathcal{R}_Q \mathcal{R}_Q^\top P
        & (\Xi_+ -  D) \hat G \\
        \star
        & -I_{n_f-n}
        \end{bmatrix}
        \preceq 0 .
    \end{equation}
    In addition, a second application of the Schur complement gives
    \begin{equation}\label{eq:lmi_schur2}
    \widetilde{\Upsilon}(Y,\alpha) + P \mathcal{R}_Q \mathcal{R}_Q^\top P + (\Xi_+ -  D)\hat G \big((\Xi_+ -  D)\hat G\big)^\top  \preceq 0 .  
    \end{equation}
    Left- and right-multiplying \eqref{eq:lmi_schur2} by $P^{-1}$ and recalling $\bar{G} := Y P^{-1}$ yield
    \begin{equation}\label{eq:lmi_transformed}
        \begin{aligned}
        &P^{-1} (\Xi_+ -  D) \bar G
        + ((\Xi_+ -  D) \bar G)^\top P^{-1}
        + \alpha P^{-2} \\
        &+ \mathcal{R}_Q \mathcal{R}_Q^\top
        + P^{-1}(\Xi_+ - D)\hat G ((\Xi_+ - D)\hat G)^\top P^{-1}
        \preceq 0.
        \end{aligned}
    \end{equation}
    Then, by Petersen’s lemma, we have that 
    \begin{equation}
        \begin{split}
          &P^{-1} (\Xi_+ -  D) \bar G + ((\Xi_+ -  D) \bar G)^\top P^{-1} + \alpha P^{-2}\\
          & + P^{-1} (\Xi_+-D) \hat G \widetilde{\mathcal{R}}
           + \widetilde{\mathcal{R}}^\top ((\Xi_+ -D)\hat G)^\top P^{-1}   \preceq 0
            \end{split}
    \end{equation}
    holds for all $\widetilde{\mathcal{R}}\in\mathcal{R}:=\{\widetilde{\mathcal{R}}: \widetilde{\mathcal{R}}\widetilde{\mathcal{R}}^\top\preceq  \mathcal R_Q \mathcal R_Q^\top\}$.
    
    From \eqref{eq:nonlin:bound}, we have
  \begin{equation*}
    \frac{\partial Q}{\partial (x,z)}^\top\frac{\partial Q}{\partial (x,z)}\preceq \begin{bmatrix}
        R_QR_Q^\top & 0\\
        0&0
    \end{bmatrix}=\mathcal{R}_Q \mathcal{R}_Q^\top
\end{equation*}
    Hence,
\begin{equation}
\begin{aligned}
P^{-1} &(\Xi_+ -  D) \bar G
+ \big((\Xi_+ -  D) \bar G\big)^\top P^{-1}
+ \alpha P^{-2} \\
&+ P^{-1} (\Xi_+ -  D)\hat G \frac{\partial Q}{\partial (x,z)}\\
&+ \frac{\partial Q}{\partial (x,z)}^\top \big((\Xi_+ -  D) \hat G\big)^\top P^{-1}
\preceq 0.
\end{aligned}
\end{equation}

Let $\beta = \alpha \lambda_{\min}(P^{-1})$ and recall $G=\begin{bmatrix}\bar G & \hat G\end{bmatrix}$. Then
\begin{equation}\label{eq:lyapunov}
    \begin{split}
       & P^{-1} (\Xi_+ -  D) G \frac{\partial \mathcal F}{\partial (x,z)}\\
        &\qquad+ \frac{\partial \mathcal F}{\partial (x,z)}^\top
        ((\Xi_+ -  D) G)^\top P^{-1}
        \preceq -\beta P^{-1}.
    \end{split}
\end{equation}
Using \eqref{eq:xiG}, this is equivalent to
\begin{equation}\label{eq:contraction}
    \begin{split}
    &\Big( (A_\xi + B_\xi \mathcal K)\frac{\partial\mathcal{F}}{\partial(x,z)} \Big)^\top P^{-1}\\
&\qquad+ P^{-1}\Big( (A_\xi + B_\xi \mathcal K)\frac{\partial\mathcal{F}}{\partial(x,z)} \Big)
\preceq -\beta P^{-1}.
    \end{split} 
\end{equation}
This establishes exponential contraction with respect to the metric induced by $P^{-1}$, and hence the closed-loop system with $\mathcal K$ defined as \eqref{eq:K} is exponentially contractive on $\mathcal{X}\times \R^{n_z}$.
\end{proof}

\begin{remark}
Note that feasibility of the SDP~\eqref{eq:sdp} typically requires that $\Xi_{-}$ has full row rank, since the equality constraints \eqref{eq:sdp:1}--\eqref{eq:sdp:2} must be satisfied.
Conversely, the full row rank condition of the stacked matrix $\left[\begin{smallmatrix}
        U_-\\
        \Xi_-\\
        M_-
        \end{smallmatrix}\right]$ provides a sufficient condition for feasibility.
This condition, however, is not necessary, and feasibility of \eqref{eq:sdp} may still hold even when it is not satisfied.
Therefore, the feasibility of \eqref{eq:sdp} itself serves as an implicit certificate that the collected data are sufficiently informative for controller synthesis.
\end{remark}

Lemma~\ref{lem:contraction} shows that the feedback gain $\K$ designed from data renders the closed-loop system \eqref{eq:closeloop:model} exponentially contractive. 
We next establish that the closed-loop system admits a unique bounded entire solution, namely a solution defined on the whole time axis $t\in\mathbb{R}$, which is almost periodic and to which all solutions converge exponentially.

\begin{lemma}\label{lem:AP_ss}
Consider the closed-loop dynamics \eqref{eq:closeloop:model} and define $\xi:=\mathrm{col}(x,z)$. Suppose Assumption~\ref{as:exo:non} holds.
Let the feedback gain $\K$ be designed as \eqref{eq:K} according to Lemma~\ref{lem:contraction}. Assume that the closed-loop solutions considered remain in the operating set on which the contraction inequality holds. Then, the closed-loop system admits a unique bounded entire solution $\xi_v^*$, which is almost periodic, and all solutions initialized in the contraction region converge to it exponentially.
\end{lemma}

\begin{proof}
First, rewrite the closed-loop dynamics \eqref{eq:closeloop:model} as
\begin{equation}\label{eq:xi_dyn}
    \dot \xi = f(\xi,v),
\end{equation}
where $f(\xi,v):=(A_\xi+B_\xi \K)\mathcal F(\xi)+E_\xi v.$
By Lemma~\ref{lem:contraction}, the closed-loop system is exponentially contractive. It is noted that $f$ is continuously differentiable in $\xi$ and continuous in $v$. Under Assumption~\ref{as:exo:non}, the exosignal $v(t)$ exists and is bounded for all $t\in\mathbb R$.
Therefore, by \cite[Theorem~1]{Pavlov2005}, the system \eqref{eq:xi_dyn} is exponentially convergent and input-to-state convergent. 
Furthermore, according to Definitions~1-3 and Property~1 in \cite{Pavlov2005}, there exists a unique solution $\xi_v^*(t)$ defined and bounded for all $t\in\mathbb R$, to which all solutions of \eqref{eq:xi_dyn} converge exponentially.

We next show that the bounded entire solution $\xi_v^*$ is almost periodic.
For any $\theta\in\mathbb R$, define the shifted input $v_\theta(t):=v(t+\theta).$
For the shifted input \(v_\theta\), the time-shifted trajectory \(\xi_v^*(t+\theta)\) is a bounded entire solution of \eqref{eq:xi_dyn}.
By uniqueness of bounded entire solutions, it follows that
\begin{equation}\label{eq:shift_identity}
\xi_{v_\theta}^*(t)=\xi_v^*(t+\theta),\quad \forall t\in\mathbb R.
\end{equation}

In addition, since the closed-loop system \eqref{eq:xi_dyn} is input-to-state convergent, by \cite[Definition~3]{Pavlov2005} there exist a $\mathcal{K}_\infty$-function $\gamma(\cdot)$~\footnote{
A function $\gamma:[0,\infty)\to[0,\infty)$ is said to be of class
$\mathcal K$ if it is continuous, strictly increasing, and $\gamma(0)=0$.
A function $\gamma:[0,\infty)\to[0,\infty)$ is said to be of class
$\mathcal K_\infty$
if it is of class $\mathcal K$ and also unbounded.} and a $\mathcal{KL}$-function $\beta(\cdot,\cdot)$~\footnote{
A function $\beta:[0,\infty)\times[0,\infty)\to[0,\infty)$ is said to be
of class $\mathcal{KL}$ if $\beta(\cdot,t)$ is of class $\mathcal K$
for each fixed $t\ge0$ and $\beta(s,t)$ decreases to $0$ as
$t\to\infty$ for each fixed $s\ge0$.
} such that, for the inputs $v(\cdot)$ and $v_\theta(\cdot)$, the corresponding solutions satisfy, for all $t\ge t_0$,
\begin{equation}\label{eq:isc_pair_theta}
\begin{split}
\|\xi_{v_\theta}^*(t)-\xi_v^*(t)\|
\le\;&
\beta \left(\|\xi_{v_\theta}^*(t_0)-\xi_v^*(t_0)\|,\,t-t_0\right)\\
&+\gamma \Big(\sup_{t_0\le s\le t}\|v(s+\theta)-v(s)\|\Big).
\end{split}
\end{equation}

Since both $\xi_{v_\theta}^*(\cdot)$ and $\xi_v^*(\cdot)$ are bounded entire solutions, for each fixed $\theta$ there exists a constant $M_\theta<\infty$ such that
\[
\|\xi_{v_\theta}^*(t_0)-\xi_v^*(t_0)\|\le M_\theta,\quad \forall\, t_0\in\mathbb R.
\]
Letting $t_0\to-\infty$ in \eqref{eq:isc_pair_theta} yields
\begin{equation}\label{eq:isc_shift_uniform}
\|\xi_{v_\theta}^*(t)-\xi_v^*(t)\|
\le
\gamma\Big(\sup_{s\in\mathbb R}\|v(s+\theta)-v(s)\|\Big),\quad \forall\, t\in\mathbb R,
\end{equation}
where we used the fact that as $t_0\to-\infty$,
\[
\beta\!\left(\|\xi_{v_\theta}^*(t_0)-\xi_v^*(t_0)\|,\,t-t_0\right)
\le
\beta(M_\theta,\,t-t_0)\to 0.
\]

Then, combining \eqref{eq:isc_shift_uniform} with \eqref{eq:shift_identity}, we obtain
\begin{equation}\label{eq:ap_bound_bohr}
\sup_{t\in\mathbb R}\|\xi_v^*(t+\theta)-\xi_v^*(t)\|
\le
\gamma\Big(\sup_{t\in\mathbb R}\|v(t+\theta)-v(t)\|\Big).
\end{equation}

Fix any $\varepsilon>0$. Since $\gamma\in\mathcal K_\infty$, there exists $\delta>0$ such that $\gamma(\delta)<\varepsilon$. Since $v$ is almost periodic under Assumption~\ref{as:exo:non}, by Definition~\ref{def:apf}, for any $\delta>0$ there exists a relatively dense set of $\theta\in\mathbb R$ such that
\[
\sup_{t\in\mathbb R}\|v(t+\theta)-v(t)\|<\delta.
\]
For any such $\theta$, it follows from \eqref{eq:ap_bound_bohr} that
\[
\sup_{t\in\mathbb R}\|\xi_v^*(t+\theta)-\xi_v^*(t)\|
\le \gamma(\delta)<\varepsilon.
\]
Hence, the set of $\varepsilon$-almost periods of $\xi_v^*$ is relatively dense in $\mathbb R$ for every $\varepsilon>0$. Therefore, $\xi_v^*$ is almost periodic, which completes the proof.
\end{proof}

\subsection{Stability and Regulation Analysis}

Under the proposed data-driven regulator, the closed-loop system admits a unique bounded entire solution $\xi_v^*$. Based on this steady-state behavior, we proceed to analyze the corresponding regulation error and its frequency-domain characterization,  which will lead to the solution of Problem~\ref{pro}.

For periodic exosignals, boundedness of the internal model state together with convergence to a periodic steady state implies elimination of the embedded harmonics; see, e.g., \cite[Proposition~3]{Astolfi2015}.
Such arguments, however, rely on the finite spectral structure of periodic signals and do not directly extend to almost periodic signals, whose spectra may be more general.
To address this issue, we introduce the following auxiliary lemma based on Fourier--Bohr analysis, which will be instrumental in characterizing the  regulation error in the almost periodic setting.

\begin{lemma}\label{lem:AP_Lem3}
Let $\eta:\mathbb{R}\to\mathbb{R}$ be almost periodic, and let
\[
\mathcal{M}\{\eta\}:=\lim_{T\to\infty}\frac{1}{T}\int_{0}^{T} \eta(t)\,dt
\]
denote its mean value. Then the following statements hold:
\begin{enumerate}
\item[i)] Consider
\begin{equation}\label{eq:AP_C7}
\dot\zeta=\phi_k\zeta+\bar G_k \eta,\quad 
\phi_k=\begin{bmatrix}0&\sigma_k\\-\sigma_k&0\end{bmatrix}
\end{equation}
where $\bar G_k\in\mathbb{R}^2\setminus\{0\}$ and $\sigma_k>0$ is a given frequency with $k\in\{1,\ldots,s\}$.
If \eqref{eq:AP_C7} admits a bounded entire solution $\zeta^*$, then
\begin{equation}\label{eq:AP_C8}
\hat \eta(\sigma_k)=\lim_{T\to\infty}\frac{1}{T}\int_0^T \eta(t)e^{-i\sigma_k t}dt=0.
\end{equation}

\item[ii)] Consider
\begin{equation}\label{eq:AP_C9}
\dot\zeta=\gamma \eta,\quad \gamma\in\mathbb{R}\setminus\{0\}.
\end{equation}
If \eqref{eq:AP_C9} admits a bounded entire solution $\zeta^*$, then
\begin{equation}\label{eq:AP_C10}
\hat \eta(0)=\mathcal{M}\{\eta\}=0.
\end{equation}

\item[iii)] Suppose $\eta$ is differentiable, $\dot \eta$ is a continuous almost periodic signal, and $|\dot \eta(t)|\le \bar \eta$ for all $t\in\mathbb{R}$.
Assume that the Fourier--Bohr spectrum of $\eta$ is denoted by
$\Lambda_\eta$.
If $\hat\eta(0)=0$ and $\hat\eta(\sigma_k)=0$ for $k=1,\ldots,s$,
and if there exists $\underline{\lambda} > 0$ such that
\[
|\lambda| \ge \underline{\lambda}, \quad
\forall \lambda \in \Lambda_\eta \setminus \{0,\sigma_1,\ldots,\sigma_s\},
\]
then
\begin{equation}\label{eq:AP_L2_bound}
\lim_{T\to\infty}\frac{1}{T}\int_{0}^{T} |\eta(t)|^2dt
\le \frac{\bar \eta^2}{\underline{\lambda}^2}.
\end{equation}
\end{enumerate}
\end{lemma}


\begin{proof}
\textit{i)} Let $\zeta^*=[\zeta_1^*,\,\zeta_2^*]^\top$ and define the complex variable
\[
\zeta_c:=\zeta_1^*-i\zeta_2^*.
\]

Writing $\bar G_k={\rm col}(\bar G_{k,1},\bar G_{k,2})$, we obtain
\begin{equation}\label{eq:complex_osc}
\dot{\zeta}_c=i\sigma_k\zeta_c+\hat G_k\eta,
\end{equation}
where $\hat G_k:=\bar G_{k,1}-i\bar G_{k,2}.$
Since $\bar G_k\in\mathbb R^2\setminus\{0\}$, we have $\hat G_k\neq0$.
Multiplying both sides of \eqref{eq:complex_osc} by $e^{-i\sigma_k t}$ yields
\[
\frac{d}{dt}\big(e^{-i\sigma_k t}\zeta_c\big)
=
\hat G_k e^{-i\sigma_k t}\eta.
\]
Integrating over $[0,T]$ gives
\[
e^{-i\sigma_k T}\zeta_c(T)-\zeta_c(0)
=
\hat G_k\int_0^T \eta(t)e^{-i\sigma_k t}\,dt.
\]
Dividing both sides by $T\hat G_k$ leads to
\begin{equation}\label{eq:divide_T}
    \frac{1}{T}\int_0^T \eta(t)e^{-i\sigma_k t}\,dt
=
\frac{e^{-i\sigma_k T}\zeta_c(T)-\zeta_c(0)}
{\hat G_k T}.
\end{equation}
Since $\zeta_c(t)$ is bounded on $\mathbb R$ and $|e^{-i\sigma_k T}|=1$, the numerator on the right-hand side of \eqref{eq:divide_T} is bounded. Hence,  the right-hand side of \eqref{eq:divide_T} converges to zero as $T\to\infty$, i.e.,
\[
\lim_{T\to\infty}\frac{1}{T}\int_0^T e^{-i\sigma_k t}\eta(t)dt=0,
\]
which proves \eqref{eq:AP_C8}.

\textit{ii)} Integrating \eqref{eq:AP_C9} over $[0,T]$ yields
\[
\zeta^*(T)=\zeta^*(0)+\gamma\int_0^T \eta(t)\,dt,
\]
which implies
\[
\frac{1}{T}\int_0^T \eta(t)\,dt=\frac{\zeta^*(T)-\zeta^*(0)}{\gamma T}.
\]
Since $\zeta^*(t)$ is bounded on $\mathbb{R}$, the right-hand side converges to zero as $T\to\infty$, and thus $\mathcal{M}\{\eta\}=0$.

\textit{iii)} Let $\Lambda_\eta$ denote the Fourier--Bohr spectrum of $\eta$.
By Parseval's identity in Lemma~\ref{lem:parseval},
\begin{equation}\label{eq:AP_Parseval_eta}
\mathcal M\{|\eta(t)|^2\}
=
\sum_{\lambda\in\Lambda_\eta}|\hat\eta(\lambda)|^2.
\end{equation}
Since $\eta$ is differentiable and $\dot\eta$ is almost periodic, the Fourier--Bohr coefficients satisfy $\hat{\dot\eta}(\lambda)=i\lambda\hat\eta(\lambda),\ \lambda\in\Lambda_\eta .$
Applying Parseval's identity to $\dot\eta$ gives
\begin{equation}\label{eq:AP_Parseval_deta}
\mathcal M\{|\dot\eta(t)|^2\}
=
\sum_{\lambda\in\Lambda_\eta}
\lambda^2|\hat\eta(\lambda)|^2.
\end{equation}

By assumption, $\hat\eta(0)=0$ and $\hat\eta(\sigma_k)=0$ for
$k=1,\ldots,s$. Hence, the only nonzero contributions to
\eqref{eq:AP_Parseval_eta} come from $\Lambda_\eta\setminus\{0,\sigma_1,\ldots,\sigma_s\}.$
By the definition of $\underline\lambda$, for all
$\lambda\in\Lambda_\eta\setminus\{0,\sigma_1,\ldots,\sigma_s\}$,
we have $|\lambda|\ge \underline\lambda$. Therefore,
\begin{align*}
\mathcal M\{|\dot\eta(t)|^2\}
&=
\sum_{\lambda\in\Lambda_\eta\setminus\{0,\sigma_1,\ldots,\sigma_s\}}
\lambda^2|\hat\eta(\lambda)|^2  \\
&\ge
\underline\lambda^2
\sum_{\lambda\in\Lambda_\eta\setminus\{0,\sigma_1,\ldots,\sigma_s\}}
|\hat\eta(\lambda)|^2  \\
&=
\underline\lambda^2\,\mathcal M\{|\eta(t)|^2\}.
\end{align*}
It follows that
\begin{equation}\label{eq:AP_step_ratio}
\mathcal M\{|\eta(t)|^2\}
\le
\frac{\mathcal M\{|\dot\eta(t)|^2\}}{\underline\lambda^2}.
\end{equation}
Finally, using $|\dot\eta(t)|\le\bar\eta$, we obtain
\[
\mathcal M\{|\dot\eta(t)|^2\}
=
\lim_{T\to\infty}\frac{1}{T}\int_0^T|\dot\eta(t)|^2\,dt
\le
\bar\eta^2.
\]
Substituting this into \eqref{eq:AP_step_ratio} yields
\[
\mathcal M\{|\eta(t)|^2\}
\le
\frac{\bar\eta^2}{\underline\lambda^2},
\]
which proves \eqref{eq:AP_L2_bound}.
\end{proof}


Lemma~\ref{lem:AP_Lem3} links the boundedness of internal-model states to the vanishing of selected Fourier--Bohr coefficients and provides a bound on the residual signal energy. Combining this result with Lemmas~\ref{lem:contraction}-\ref{lem:AP_ss}, we obtain the following theorem, which solves Problem~\ref{pro}.

\begin{theorem}\label{thm:main_prob1}
Consider the system \eqref{eq:sys:non} under Assumption~\ref{as:exo:non}.
Suppose that the conditions of Lemma~\ref{lem:contraction} hold, and
let the controller \eqref{eq:ctrl} be implemented with $\K$ designed as in \eqref{eq:K}.
Then, for every closed-loop trajectory remaining in the certified contraction region, the solution is bounded for all $t\ge 0$, and the regulation error $e(t)$ converges exponentially to a bounded almost periodic signal $e^*(t)$.
Moreover, the Fourier--Bohr coefficients of $e^*(t)$ at frequencies $0$ and $\{\sigma_k\}_{k=1}^s$ vanish, namely,
\[
\hat e_j^*(0)=0,\quad \hat e_j^*(\sigma_k)=0,\quad k=1,\ldots,s,\; j=1,\ldots,p.
\]
Furthermore, there exist constants $\bar e_j>0$, $j=1,\ldots,p$, such that the time-averaged energy of each component $e_j^*(t)$ satisfies
\begin{equation}\label{eq:e_s+1}
\lim_{T\to\infty}\frac{1}{T}\int_0^T |e^*_j(t)|^2\,dt
\le \frac{\bar e_j^2}{\underline{\lambda}^2},
\end{equation}
where $\underline{\lambda} > 0$ denotes the smallest magnitude of the frequencies in the Fourier--Bohr spectrum of $e_j^*(t)$ that are not contained in $\{0,\sigma_1,\ldots,\sigma_s\}$.
\end{theorem}

\begin{proof}
By Lemma~\ref{lem:contraction} and Lemma~\ref{lem:AP_ss}, the closed-loop system \eqref{eq:closeloop:model} admits a unique bounded entire solution $\xi^*(t)={\rm col}(x^*(t),z^*(t))$ defined for all $t\in\R$, which is almost periodic and to which all solutions converge exponentially.

Since $\xi^*$ and $v$ are almost periodic and the nonlinear function $F(x)$ is continuous, it follows from \eqref{eq:sys:non:e} that $e^*$ is almost periodic. 
Hence, $e(t)$ converges exponentially to the almost periodic signal $e^*(t)$.

Next, consider the steady-state internal model
\[
\dot z^*=\Phi z^*+\Gamma e^*.
\]
By the structure of $(\Phi,\Gamma)$, each component $e_j^*$, $j=1,\ldots,p$, enters an internal model consisting of one integrator and $s$ oscillatory subsystems with frequencies $\{\sigma_k\}_{k=1}^s$. 
Since $z^*(t)$ is bounded on $\R$, it follows from Lemma~\ref{lem:AP_Lem3} i) and ii) that
\[
\hat e_j^*(0)=0,\quad \hat e_j^*(\sigma_k)=0,\quad k=1,\ldots,s,
\]
for each $j=1,\ldots,p$.

Moreover, differentiating \eqref{eq:sys:non:e} along $\xi^*$ gives
\begin{equation}
    \dot e^*=C\frac{\partial F}{\partial x}(x^*)\dot x^*+HSv.
\end{equation}
Since $x^*$ and $v$ are almost periodic, and $F$ and $\frac{\partial F}{\partial x}$ are continuous on the bounded trajectory, it follows that both $\dot x^*$ and $\dot e^*$ are almost periodic. Moreover, boundedness of $x^*$ and $v$ implies that $\dot e^*$ is bounded on $\mathbb{R}$.
Hence, for each $j=1,\ldots,p$, there exists $\bar e_j>0$ such that $|\dot e_j^*(t)|\le \bar e_j,\ \forall t\in\R.$
Therefore, applying Lemma~\ref{lem:AP_Lem3} iii) to each component $e_j^*(t)$ yields \eqref{eq:e_s+1}.
This completes the proof.
\end{proof}

\begin{remark}
We emphasize that the regulation error does not converge to zero, as the steady-state response may contain frequency components that are not captured by the internal model.
Specifically, the proposed approach ensures that the Fourier--Bohr coefficients of the steady-state error at the embedded frequencies vanish. The remaining
components, corresponding to frequencies outside the embedded set, constitute a residual almost periodic error. The bound in \eqref{eq:e_s+1} quantifies the energy of this residual in terms of the
smallest magnitude of the unmodeled frequencies.
\end{remark}

\begin{remark}
Several data-driven output regulation methods have been reported in the literature. Compared with these existing results, the proposed approach exhibits the following distinguishing features.
\begin{itemize}
    \item [i)]  \emph{A broader class of exogenous signals.} The results in \cite{hu2024output} focus on periodic exosignals, which require the disturbance or reference signals to satisfy rational frequency relationships. In contrast, the proposed framework accommodates almost periodic exosignals, allowing for more general signal compositions without imposing periodicity assumptions.

 \item [ii)]  \emph{Nonlinear system setting.} 
  Most existing data-driven output regulation methods \cite{Zhu2024,Jiao2021,lin2024ddimp} are developed for linear systems, whereas our proposed approach directly addresses a class of nonlinear systems without relying on local linearization, which allows for a wider range of application scenarios.

 \item [iii)]  \emph{Regulation objective and solution design.} The method in \cite{Liu2025regulation} studies the $k$th-order nonlinear output regulation problem by constructing a $k$-fold internal model, which achieves approximate regulation characterized by higher-order error terms. In contrast, the proposed approach addresses the original nonlinear system through selected-frequency internal models and Fourier--Bohr analysis, leading to vanishing embedded-frequency coefficients and a residual energy bound.
\end{itemize}
\end{remark}

\section{Simulation}\label{sec:simu}

In this section, the proposed data-driven output regulation framework is validated through both numerical and physics-based simulations on a quadrotor UAV subject to oscillatory disturbances. 
The numerical experiments are conducted on the nominal system model, while the physics-based simulations are implemented in the \texttt{gym-pybullet-drones} environment~\cite{Panerati2021}, which provides a realistic platform for quadrotor dynamics.

\subsection{Simulation Setup}
We consider a quadrotor UAV carrying a cable-suspended payload and focus on the hovering regulation problem in the presence of payload-induced oscillatory disturbances. 
Such scenarios arise in aerial transportation and manipulation tasks, where payload swing induces persistent oscillatory disturbances whose amplitudes and phases depend on unknown initial conditions, while their dominant frequencies are determined by physical parameters such as the cable length.
The payload swing generates oscillatory disturbances that are transmitted to the quadrotor through the cable constraints, affecting both the vertical force and the attitude dynamics.

Accordingly, the quadrotor dynamics are described by 
\begin{subequations}
    \begin{align}\label{eq:quad_dyn}
       & \dot p_z  = v_z \\
       & m\dot v_z = -mg + \sum_{i=1}^{4}f_i + d_f\\
       &\dot \eta =T(\phi,\theta)\omega\\
       & J\dot \omega + \omega\times(J\omega) = M + d_\tau
    \end{align}
\end{subequations}
where $p_z\in\mathbb{R}$  denotes the vertical regulation error with respect to the desired hovering height,
$v_z\in\mathbb{R}$ is the vertical velocity, $m$ is the mass, $g$ is the gravitational acceleration, and $f_i$, $i=1,\ldots,4$, are the thrust forces generated by the propellers. Moreover, $\omega=\mathrm{col}(p,q,r)\in\mathbb{R}^3$ is the body-frame angular velocity, $M\in\mathbb{R}^3$ is the total control moment, $\eta=\mathrm{col}(\phi,\theta,\psi)$ denotes the Euler angles, and $J=\mathrm{blkdiag}(J_x,J_y,J_z)$ is the inertia matrix. The terms $d_f$ and $d_\tau$ denote external disturbances acting on the translational and rotational dynamics.

Next, we introduce the following exosystem:
\begin{subequations}
    \begin{align}
           \dot v(t) &= \begin{bmatrix}
        0 &0& 0&0& 0\\
        0 & 0&  \sigma_z&0& 0\\
        0 & -\sigma_z&  0&0& 0\\
         0 &0& 0&0& \sigma_\omega\\
          0 &0& 0& -\sigma_\omega& 0
    \end{bmatrix}v(t)=:Sv(t)\\
        d(t)&=  \begin{bmatrix}
        0& 0& 0 & 0& 0\\
        -g & 1 & 0& 0& 0\\
         0& 0& 0& 1& 0\\
          0& 0& 0& 1& 0\\
           0& 0& 0& 0& 0\\
    \end{bmatrix}v(t) =: Ev(t)
    \end{align}
\end{subequations}
where $\sigma_z$ and $\sigma_\omega$ denote the dominant disturbance frequencies associated with the vertical force channel and the angular-rate dynamics, respectively.

Let $x=\mathrm{col}(p_z,v_z,\eta,\omega)\in\mathbb{R}^8$
be the system state and let $u=\mathrm{col}(u_1,u_2,u_3,u_4)$
be the control input, where $u_1=\sum_{i=1}^4 f_i$ is the collective thrust. The remaining inputs are associated with the rotational dynamics and are mapped to the body-frame $M=\operatorname{col}(M_1,M_2,M_3)$ by
\[
\begin{bmatrix}
M_1\\ M_2\\ M_3
\end{bmatrix}
=
\begin{bmatrix}
c_1 & 0 & 0\\
0 & c_2 & 0\\
0 & 0 & c_3
\end{bmatrix}
\begin{bmatrix}
u_2\\ u_3\\ u_4
\end{bmatrix},
\]
where $c_1,c_2,c_3$ are unknown control effectiveness coefficients.
In particular, $u_2,u_3,u_4$ are treated as virtual control variables associated with the rotational channels.

Under the above modeling setup, we select the known basis function as $F(x)=\mathrm{col}\bigl(x,\,Q(x)\bigr)$,
where
\(Q(x)=\mathrm{col}(
\sin x_3 \tan x_4 x_7,
\cos x_3 \tan x_4\, x_8,
(\cos x_3 - 1) x_7,
\sin x_3 x_8,
\frac{\sin x_3}{\cos x_4} x_7,
(\frac{\cos x_3}{\cos x_4} -1)x_8,
x_7 x_8,
x_8 x_6,
x_6 x_7).\)
Here, the operating set is chosen such that $|x_4|\leq \bar x_4<\pi/2$, 
which ensures that $\tan x_4$ and $1/\cos x_4$ are well defined.
With this choice, the system dynamics can be expressed in the form of \eqref{eq:sys:non}, where the matrices $A=
\begin{bmatrix}
{\bar A} & {\hat A}
\end{bmatrix}$, $B$, and $C=
\begin{bmatrix}
{\bar C} & {\hat C}
\end{bmatrix}$ are specified as follows
\begin{subequations}
    \begin{align}
\bar A & = \begin{bmatrix}
    0 & 1 & 0 & 0 & 0 & 0 & 0 & 0\\
    0 & 0 & 0 & 0 & 0 & 0 & 0 & 0\\
    0 & 0 & 0 & 0 & 0 & 1 & 0 & 0 \\
    0 & 0 & 0 & 0 & 0 & 0 & 1 & 0\\
    0 & 0 & 0 & 0 & 0 & 0 & 0 & 1\\
    0 & 0 & 0 & 0 & 0 & 0 & 0 & 0\\
    0 & 0 & 0 & 0 & 0 & 0 & 0 & 0\\
    0 & 0 & 0 & 0 & 0 & 0 & 0 & 0 
\end{bmatrix}\\
\hat A & =\begin{bmatrix}
    0 & 0 & 0 & 0 & 0 & 0 & 0 & 0 & 0\\
    0 & 0 & 0 & 0 & 0 & 0 & 0 & 0 & 0\\
    1 & 1 & 0 & 0 & 0 & 0 & 0 & 0 & 0\\
    0 & 0 & 1 & 1 & 0 & 0 & 0 & 0 & 0\\
    0 & 0 & 0 & 0 & 1 & 1 & 0 & 0 & 0\\
    0 & 0 & 0 & 0 & 0 & 0 & \frac{J_y-J_z}{J_x} & 0 & 0\\
    0 & 0 & 0 & 0 & 0 & 0 &  0 & \frac{J_z-J_x}{J_y} & 0\\
    0 & 0 & 0 & 0 & 0 & 0 &  0 & 0 & \frac{J_x-J_y}{J_z}
\end{bmatrix}\\
        B&= \begin{bmatrix}
                 0& 0& 0& 0 \\
                 \frac{1}{m} & 0& 0& 0\\
                 0& 0& 0& 0 \\
                 0& 0& 0& 0 \\
                 0& 0& 0& 0 \\
                0  & \frac{c_1}{J_x}& 0& 0\\
                0  & 0& \frac{c_2}{J_y}& 0\\
                  0  & 0& 0& \frac{c_3}{J_z}
        \end{bmatrix}\\      
      \bar C &= \begin{bmatrix}
          1 & 0 & 0 & 0 & 0 & 0 & 0 & 0\\
          0 & 0 & 0 & 1 & 0 & 0 & 0 & 0
      \end{bmatrix},\\
      \hat C & =0_{2\times 9}.
      \end{align}
\end{subequations}
We emphasize that the matrices $A$, $B$, and $C$ are displayed only to specify the simulation model; the proposed controller is synthesized directly from data and does not require their numerical values.

In this setting, the controlled output is chosen as $ y=\mathrm{col}(p_z,\theta)$. Accordingly, the regulation error is given by $e=y.$
The control objective is to design a feedback controller $u$ such that the quadrotor maintains the hovering equilibrium by driving the altitude deviation $p_z$ and the pitch angle $\theta$ to a small neighborhood of the origin despite the presence of payload-induced almost periodic disturbances.

\subsection{Numerical Experiments}\label{sec:num}

In this subsection, numerical simulations are conducted in Python using CVXPY with the MOSEK solver on a MacBook Pro equipped with an Apple M4-Pro processor to validate the proposed data-driven output regulation framework.

Let $m=m_q+m_p$, where $m_q$ and $m_p$ denote the quadrotor and payload masses, respectively.
Specifically, the physical parameters are selected as
$m_q = 1.2~\mathrm{kg}$, $m_p = 0.8~\mathrm{kg}$, 
$g = 9.8~\mathrm{N/kg}$, $J = \mathrm{diag}\{0.01668,\,0.01772,\,0.02966\}$, $\sigma_z=\pi/5$ $\sigma_\omega=1$, and $c_1 = c_2 = c_3 = 1$. 
To implement the proposed controller, the internal model is constructed according to the exosystem structure in Section~\ref{sec:result}. 
We set $\gamma = 1$ and $\bar{G}_1 = \bar{G}_2 = \mathrm{col}(0,1)$ in \eqref{eq:inter_matrices}. 
Since the regulated output is two-dimensional, a $p$-copy internal model with $p=2$ is adopted embedded in the controller. 
If we set $\mathcal X =
\left\{
x\in\mathbb R^8:
|x_i|\le \bar x_i,\ i=1,\ldots,8
\right\},$ where $\bar x =
\begin{bmatrix}
0.11 & 0.02 & 0.22 & 0.18 & 0.13 & 0.08 & 0.06 & 0.06
\end{bmatrix}^{\top},$
then \eqref{eq:nonlin:bound} is satisfied with $R_Q=0.5 I_8$.


For the data-driven controller design, noisy input–state trajectory data are collected using a sampling interval $t_{\ell+1} - t_\ell = 0.01~\mathrm{s}$, resulting in $T = 100$ samples. 
The control inputs are drawn uniformly from $[-1.5,\,1.5]$, while the exosignal and additive data noise are sampled from $[-0.001,\,0.001]$. 
The noise bound in Assumption~\ref{as:D_bound} is specified by selecting $\Theta  = 0.01 I$, which is consistent with the magnitude of the generated data noise. 
Based on the collected data, the matrices $U_-$, $\Xi_-$, and $\Xi_+$ are constructed accordingly.
In addition, the matrix $M_-$ is constructed as follows
\[
M_- = \begin{bmatrix}
1 & 1 & \cdots & 1 \\
\cos(\sigma_z t_0) & \cos(\sigma_z t_1) & \cdots & \cos(\sigma_z t_{T-1}) \\
\sin(\sigma_z t_0) & \sin(\sigma_z t_1) & \cdots & \sin(\sigma_z t_{T-1}) \\
\cos(\sigma_\omega t_0) & \cos(\sigma_\omega t_1) & \cdots & \cos(\sigma_\omega t_{T-1}) \\
\sin(\sigma_\omega t_0) & \sin(\sigma_\omega t_1) & \cdots & \sin(\sigma_\omega t_{T-1})
\end{bmatrix},
\]
where $\{t_k\}_{k=0}^{T-1}$ denotes the sampling instants.

Then, the data-dependent SDP \eqref{eq:sdp} is solved to obtain the feedback gain matrix $\K$. 
The corresponding simulation result is depicted in Fig.~\ref{fig:simu}. 
Both regulation error components converge rapidly to a small neighborhood of the origin and exhibit small residual oscillations in steady state, which is consistent with the theoretical regulation analysis developed in Section~\ref{sec:result}. 
These results confirm the effectiveness of the proposed data-driven internal model-based controller for output regulation under almost periodic disturbances.

\begin{figure}[!htb]
	\centering
	\includegraphics[width=9.2cm]{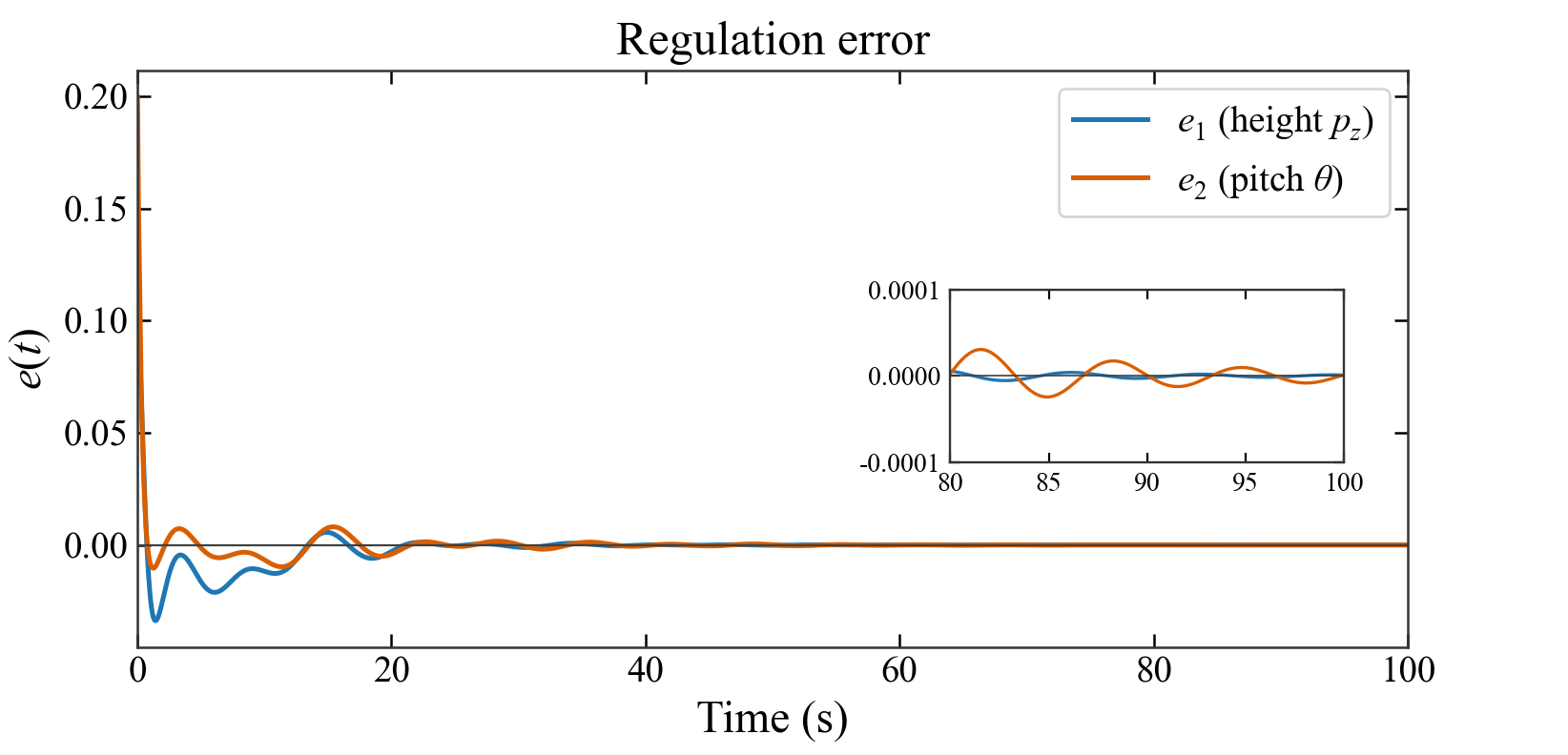}\\
	\caption{Time evolution of the regulation  error in the numerical simulation under the proposed data-driven internal model-based controller.}\label{fig:simu}
\end{figure}

\subsection{Physics-Based Experiments}
To further validate the proposed approach in a more realistic setting, physics-based experiments are conducted in the \texttt{gym-pybullet-drones} environment. 
Unlike the numerical simulations where the disturbances are explicitly generated from the exosystem, the oscillatory disturbances in this setting arise implicitly from the coupled cable--payload dynamics. 
This allows us to evaluate the effectiveness of the proposed controller under physically generated oscillatory disturbances.
The simulation environment is shown in Fig.~\ref{fig:envi}. 
A quadrotor UAV is connected to a payload via a rigid cable, forming a suspended-load system.
The motion is described with respect to an inertial coordinate frame, where the $z$-axis denotes the vertical direction and the $x$-$y$ plane corresponds to the horizontal plane.

\begin{figure}[!htb]
	\centering
	\includegraphics[width=8cm]{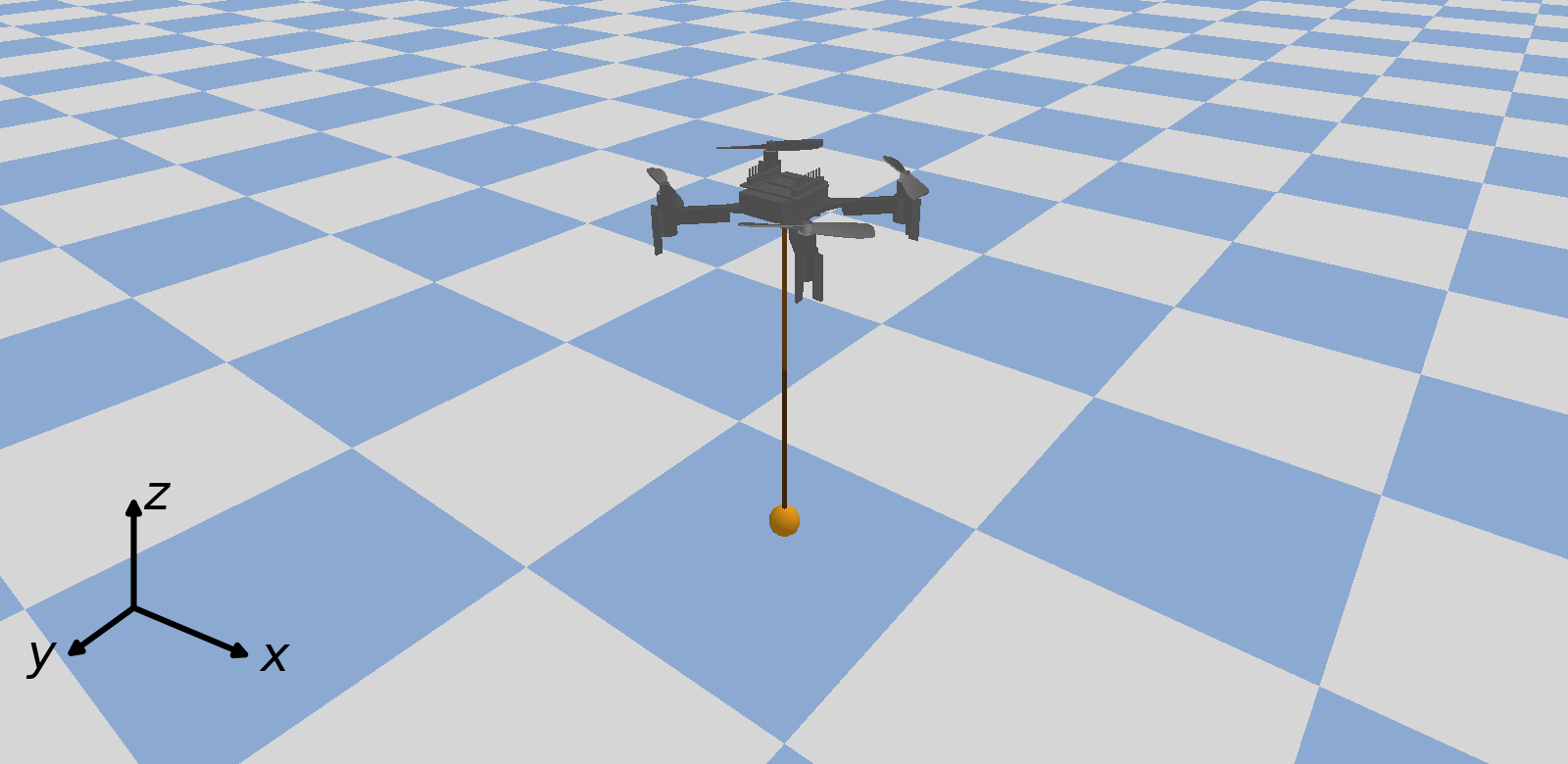}\\
	\caption{Simulation environment of the quadrotor UAV with a cable-suspended payload.}\label{fig:envi}
	\centering
\end{figure}

    \subsubsection{Implementation Details}
The simulations are carried out in the \texttt{gym-pybullet-drones} environment based on the PyBullet rigid-body physics engine. 
The quadrotor, cable, and payload are simulated as a coupled rigid-body system. 
Specifically, the cable is approximated by a $5$-link rigid-body chain connected through point-to-point joints, with its upper end attached to the quadrotor and its lower end attached to the payload. 
The payload is initialized with an angular offset from the vertical direction ($15^\circ$ in the simulations), which induces pendulum-like oscillations. 
No artificial damping is introduced into the cable--payload system, so that the oscillatory motion persists throughout the experiment. 
As a result, the cable tension generates oscillatory disturbances affecting both the translational and rotational dynamics.

Within this environment, the proposed controller is implemented as an outer-loop regulation module acting on the translational and rotational dynamics, while the built-in low-level attitude stabilization mechanism is retained as an inner loop. 
Accordingly, the controller generates collective thrust and attitude-related reference commands for hovering regulation.

The physical parameters, including $m$, $g$, $J$, and $c_1,c_2,c_3$, are selected to be the same as those used in Sec.~\ref{sec:num}. 
In this setting, the dominant oscillation frequency associated with the payload motion is determined by the cable length and incorporated into the controller design, so that the embedded frequencies reflect the dominant disturbance characteristics of the suspended-load system~\cite{Guo2020}.
Specifically, the cable length is set to $L=2.0~\mathrm{m}$, which determines the payload oscillation frequency $\sigma_z=\sqrt{g/L}$. 
In addition, the payload motion induces oscillatory torque disturbances in the rotational dynamics. 
To account for these effects, an auxiliary frequency $\sigma_\omega=1~\mathrm{rad/s}$ is introduced to capture representative oscillatory effects in the rotational channels.
It is worth emphasizing that, except for the disturbance frequencies used in the controller design, all the above parameters are used solely for simulation and data generation and they are not required in the data-driven controller synthesis.

The internal model construction and the data-driven controller synthesis settings are selected the same as those used in Sec.~\ref{sec:num}. 
In particular, the same data collection procedure and SDP-based controller design method are adopted in the physics-based experiments. 
The resulting feedback gain matrix $\K$ is then implemented in the outer-loop controller for closed-loop simulations.

    \subsubsection{Results}
We first evaluate the closed-loop performance of the proposed controller in the physics-based environment.
Fig.~\ref{fig:error} shows the time evolution of the regulation error $e(t)$ under the proposed data-driven controller. 
It can be observed that both the altitude deviation $p_z$ as well as the pitch angle $\theta$ converge rapidly to a small neighborhood of the origin and remain bounded with small oscillations, demonstrating effective regulation of the dominant oscillatory disturbances.
Furthermore, Fig.~\ref{fig:control} depicts the corresponding control inputs. 
The collective thrust quickly settles around the hover value, while the control torques remain bounded and exhibit oscillatory patterns consistent with the disturbance characteristics. 
This indicates that effective disturbance rejection is achieved without excessive control effort.

\begin{figure}[H]
	\centering
    \includegraphics[width=8.8cm]{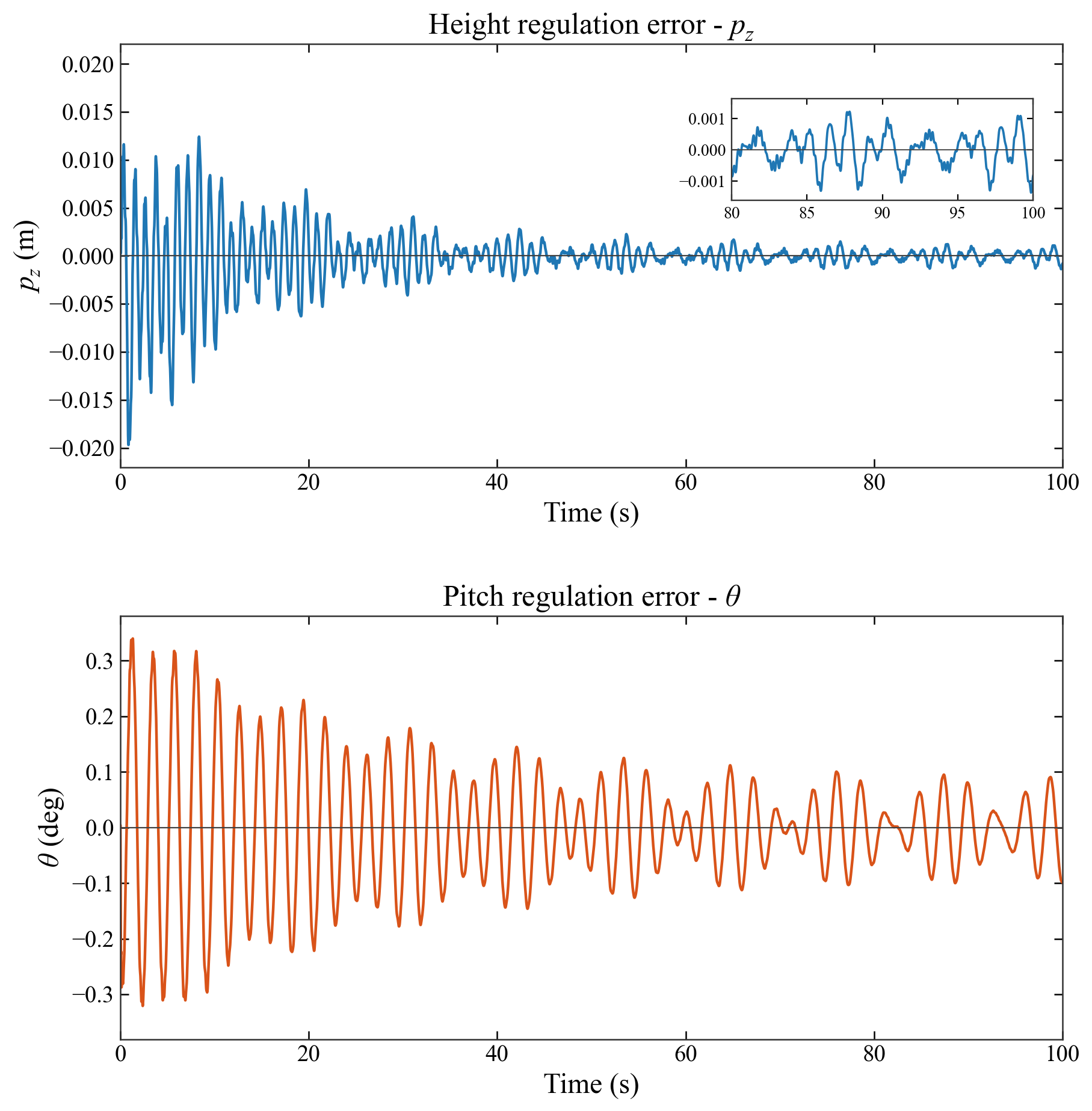}\\
	\caption{Time evolution of the regulation  error in the physics-based experiments under the proposed data-driven internal model-based controller.}\label{fig:error}
\end{figure}

\begin{figure}[H]
	\centering
	\includegraphics[width=8.8cm]{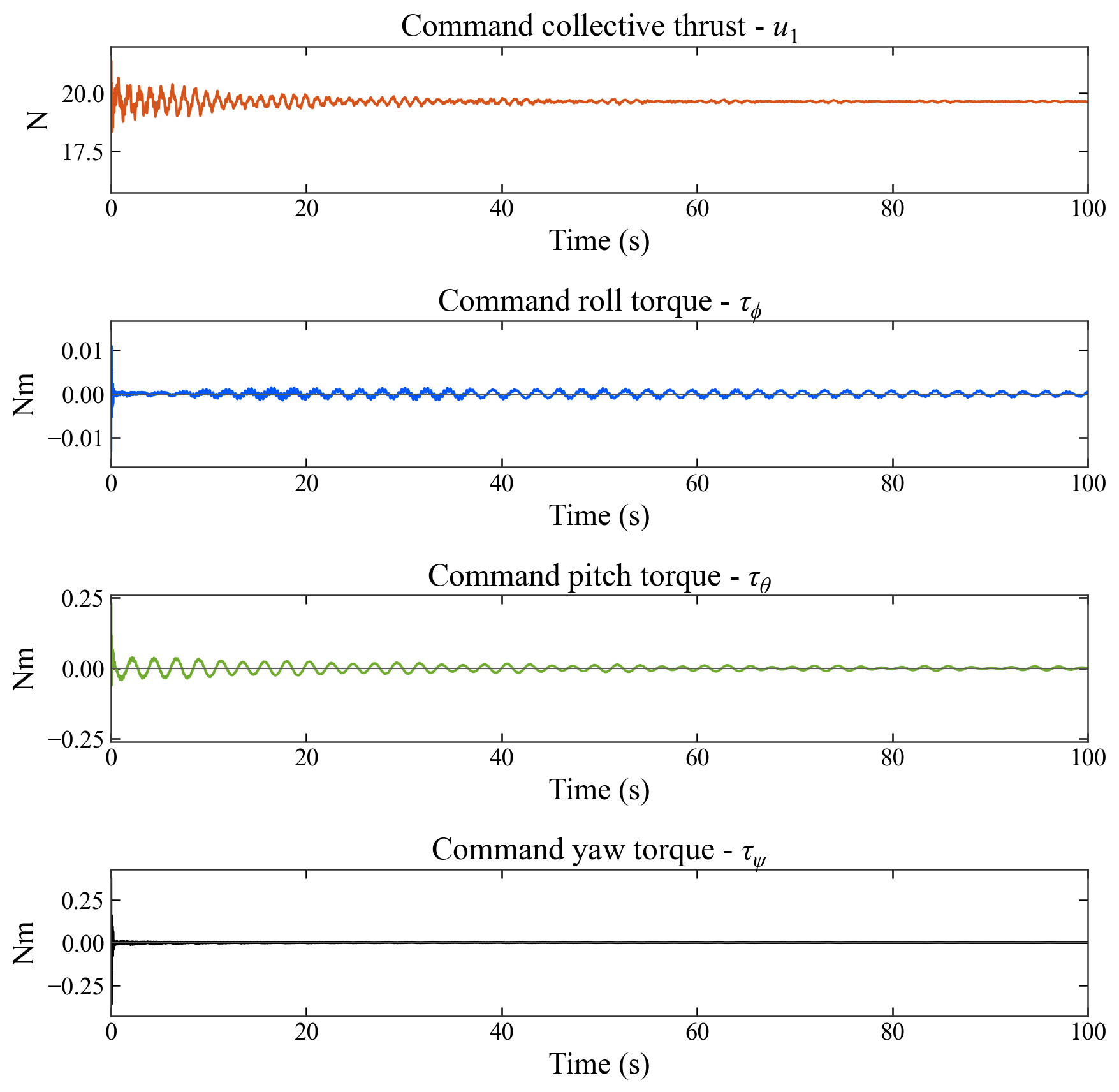}\\
	\caption{Control inputs under the proposed data-driven internal model-based controller.}\label{fig:control}
\end{figure}


\begin{figure}[tb]
	\centering
	\includegraphics[width=9.3cm]{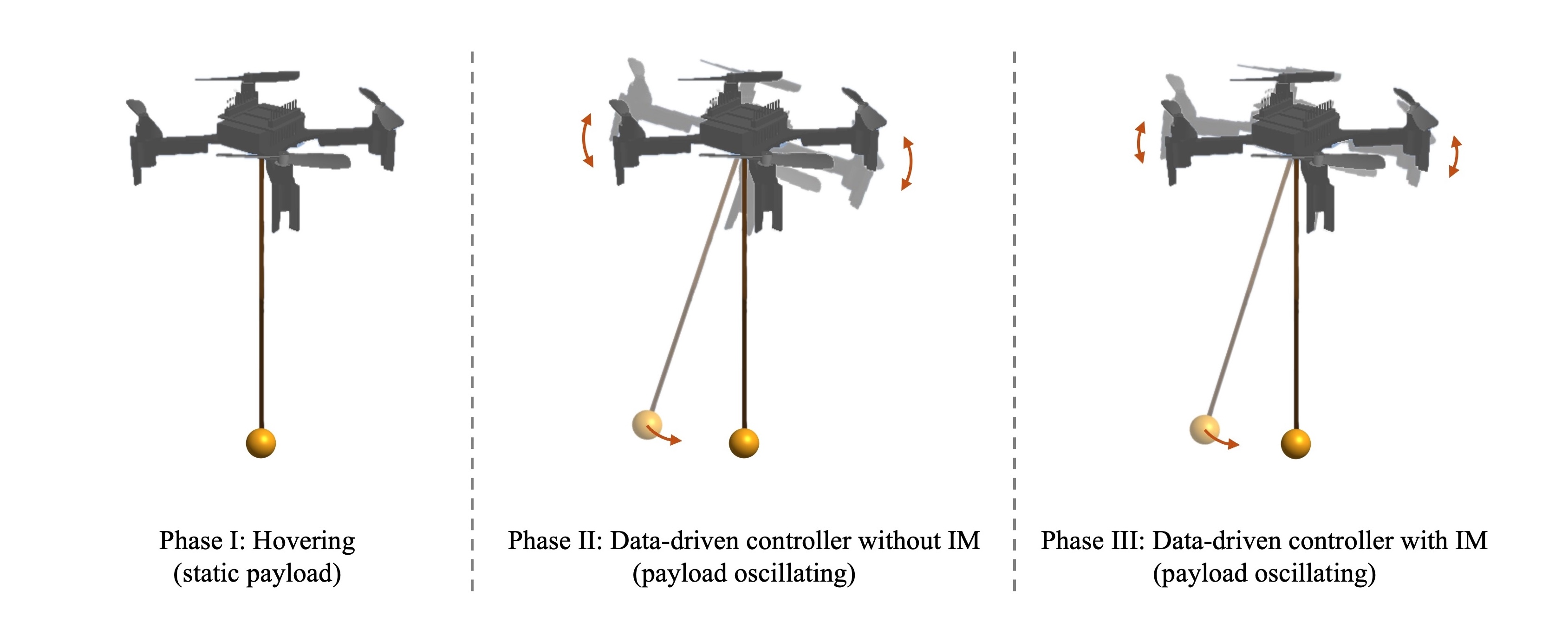}\\
	\caption{Illustration of the three-phase staged switching experiment with a cable-suspended payload.}\label{fig:phase}
\end{figure}

\begin{figure}[!htb]
	\centering
	\includegraphics[width=8.8cm]{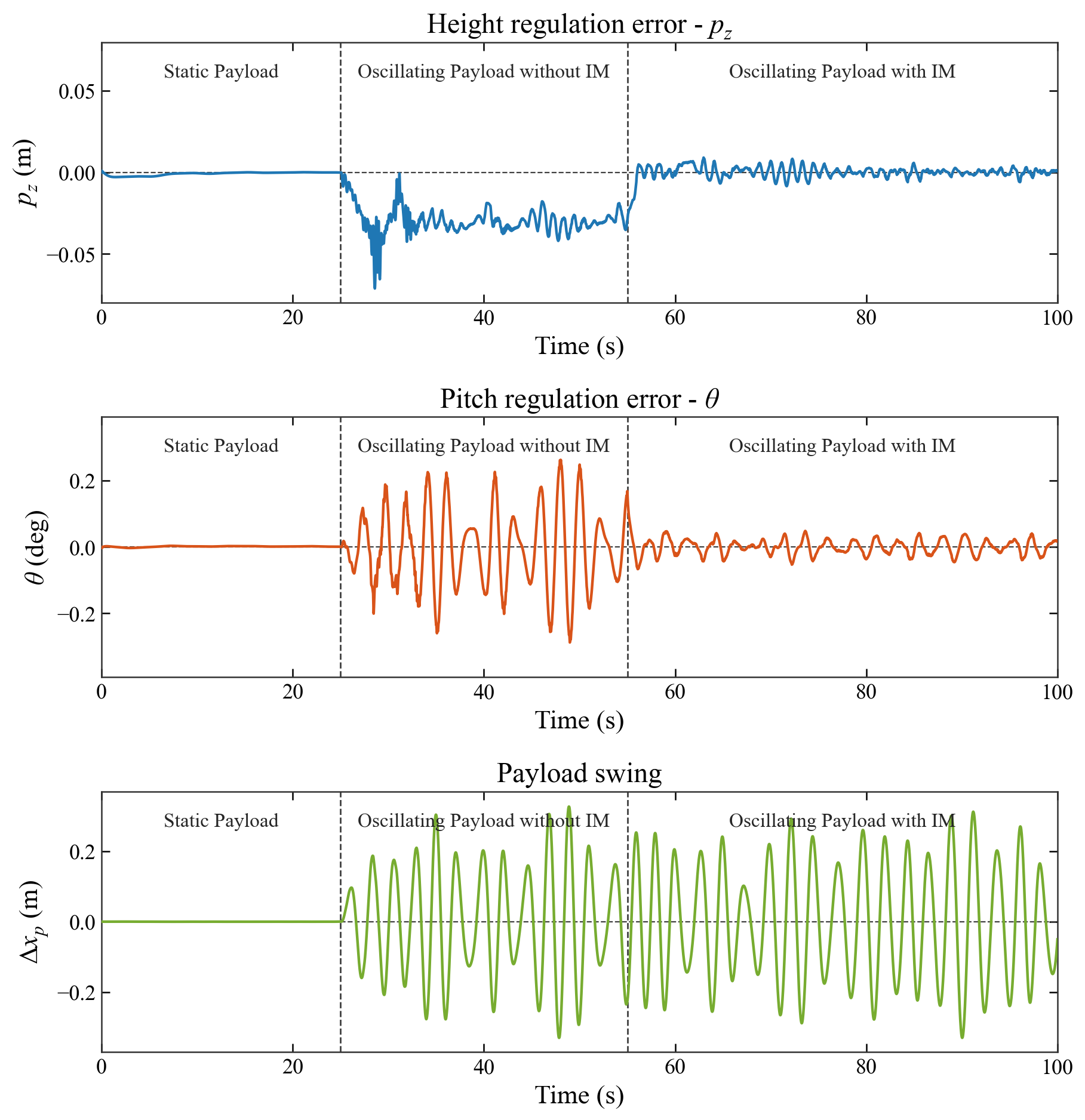}\\
	\caption{Regulation errors and payload swing under different payload conditions and staged switching between data-driven controllers without and with the internal model (IM).}\label{fig:compare}
\end{figure}


To demonstrate the effectiveness of the proposed data-driven controller in the presence of payload-induced disturbances, a staged switching experiment is conducted, as illustrated in Fig.~\ref{fig:phase}. During the flight, the payload is initially kept static and is then excited by an external horizontal force to induce persistent oscillations.
The payload swing is quantified by the relative horizontal displacement $\Delta x_p = x_p - x_q$
where $x_p$ and $x_q$ denote the positions of the payload and the quadrotor along the $x$-axis, respectively. 
The experiment consists of three phases: \emph{Phase I}, where the payload is static and the quadrotor maintains hovering; \emph{Phase II}, where the payload oscillates and a data-driven controller without the internal model is applied; and \emph{Phase III}, where the payload continues to oscillate while the proposed data-driven internal model-based controller is activated.
In particular, the controller without the internal model is designed by solving the corresponding data-driven SDP without the internal-model dynamics, resulting in a static feedback law of the form $u = \mathcal{K}_x \mathcal{F}(x)$.
The corresponding regulation errors and the payload swing are shown in Fig.~\ref{fig:compare}. During Phase I, the system maintains satisfactory hovering performance, with all signals remaining close to zero. When the payload starts to oscillate in Phase II, pronounced oscillations appear in both the height and pitch regulation errors, indicating that the controller without the internal model has limited capability in compensating for the oscillatory disturbance. In contrast, after switching to the proposed controller in Phase III, the oscillation amplitudes of both regulation errors are substantially reduced, while the payload swing remains at a comparable level. This confirms that the performance improvement is not caused by a reduction of the payload disturbance itself, but rather by the improved disturbance rejection capability introduced by the proposed controller.
Overall, the results demonstrate that the proposed data-driven internal model-based controller effectively attenuates payload-induced oscillatory disturbances and achieves improved regulation performance.

\section{Conclusion}\label{sec:conclusion}

This paper developed a direct data-driven method for frequency-selective output regulation of nonlinear systems subject to almost periodic exosignals. By embedding the prescribed exosystem frequencies in a dynamic internal model and by representing the unknown augmented dynamics through offline input--state data, the proposed approach avoids both explicit model identification and the solution of nonlinear regulator equations. A robust SDP was derived to synthesize a feedback gain from noisy data. The resulting controller certifies exponential contraction of the augmented closed-loop system on a prescribed operating set. This contraction property yields a unique attracting bounded entire solution; moreover, because the exosignal is almost periodic, the attracting steady state is also almost periodic.

The regulation result was established in the Fourier--Bohr domain. The boundedness of the internal model forces the Fourier--Bohr coefficients of the steady-state regulation error at the embedded frequencies to vanish. The remaining spectral components, which cannot in general be removed without embedding additional frequencies or solving the full nonlinear regulation problem, were bounded through a Parseval-type energy estimate. Numerical and physics-based quadrotor simulations demonstrated that the internal-model controller substantially attenuates payload-induced oscillations compared with a data-driven controller without an internal model.

Several extensions are important for future work. First, the present method requires state measurements and a known nonlinear dictionary; output-feedback versions with data-driven observers would make the approach more broadly applicable. Second, the contraction certificate is local to the prescribed operating set, and systematic enlargement of this set is an open problem. Third, adaptive or data-driven estimation of unknown dominant exosystem frequencies would remove the need for a prescribed frequency set. Finally, experimental validation on hardware and discrete-time formulations are natural next steps toward deployment.

\bibliographystyle{IEEEtran}
\bibliography{ddOutReg}

\end{document}